\documentclass[10pt,journal,compsoc]{IEEEtran}
\usepackage{url}
\usepackage{amssymb,amsmath,amsfonts,amsthm}
\usepackage{enumerate}

\usepackage{enumitem}
\usepackage{array}
\usepackage{framed}
\usepackage{multirow}
\usepackage{dblfloatfix}
\usepackage{etoolbox}
\usepackage{xcolor,xspace}
\usepackage{float,subfigure}
\usepackage{tabularx}
\usepackage[referable]{threeparttablex}
\usepackage{verbatim}
\usepackage{caption}
\usepackage{boxedminipage}
\usepackage{graphicx}
\usepackage{etoolbox}
\usepackage{setspace}
\usepackage{threeparttable}
\usepackage{dblfloatfix}
\patchcmd{\paragraph}{\itshape}{\bfseries\boldmath}{}{}
\usepackage[colorlinks,
            linkcolor=black,
            anchorcolor=black,
            citecolor=blue
            ]{hyperref}

\newcommand{\prob}{\mathbf{Pr}}
\newcommand{\N}{\mathbb{N}}
\newcommand{\punc}{\mathsf{Puncture}}

\newcommand{\expt}{\mathbf{Expt}}
\newcommand{\BF}{\mathsf{BF}}
\newcommand{\gen}{\mathsf{Gen}}
\newcommand{\update}{\mathsf{Update}}
\renewcommand{\check}{\mathsf{Check}}
\newcommand{\rsample}{\overset{\$}{\leftarrow}}
\newcommand{\G}{\mathbb{G}}
\newcommand{\Z}{\mathbb{Z}}
\newcommand{\delegate}{\mathsf{Delegate}}
\newcommand{\params}{\mathsf{params}}
\newcommand{\tabincell}[2]{\begin{tabular}{@{}#1@{}}#2\end{tabular}}
\theoremstyle{definition}
\newtheorem{definition}{\textbf{\textup{Definition}}}
\newtheorem{remark}{\textbf{\textup{Remark}}}
\newtheorem{theorem}{\textbf{\textup{Theorem}}}
\newtheorem{lemma}{\textbf{\textup{Lemma}}}[section]
\newcommand{\ignore}[1]{}

\newcommand{\secparam}{\ensuremath{\lambda}\xspace}
\newcommand{\A}{\ensuremath{\mathcal{A}}\xspace}
\newcommand{\B}{\ensuremath{\mathcal{B}}\xspace}

\newcommand{\ppt}{\ensuremath{\textsc{ppt}}\xspace}
\newcommand{\negl}{\ensuremath{\mathsf{negl}}\xspace}

\newcommand{\sk}{\ensuremath{\mathsf{sk}}\xspace}
\newcommand{\vk}{\ensuremath{\mathsf{vk}}\xspace}

\newcommand{\setup}{\ensuremath{\mathsf{Setup}}\xspace}

\newcommand{\sign}{\ensuremath{\mathsf{Sign}}\xspace}
\newcommand{\verify}{\ensuremath{\mathsf{Verify}}\xspace}

\ifCLASSOPTIONcompsoc
  \usepackage[nocompress]{cite}
\else
  \usepackage{cite}
\fi

\ifCLASSINFOpdf

\else

\fi
\begin{document}
\title{Puncturable Signatures and Applications in Proof-of-Stake Blockchain Protocol} 

\author{
{Xin-Yu Li, Jing Xu, Xiong Fan, Yu-Chen Wang and Zhen-Feng Zhang}
\IEEEcompsocitemizethanks{

\IEEEcompsocthanksitem X.-Y. Li, J. Xu, Y.-C Wang and Z.-F. Zhang are with Trusted Computing and Information Assurance Laboratory, Institute of Software, Chinese Academy of Sciences, Beijing, China.
E-mails: \{xujing, xinyu2016\}@iscas.ac.cn,\{wangyuchen, zfzhang\}@tca.iscas.ac.cn.
\IEEEcompsocthanksitem X. Fan is with  University of Maryland, College Park, MD, America.
E-mail: leofanxiong@gmail.com.}
}




\IEEEtitleabstractindextext{
\begin{abstract}
Proof-of-stake blockchain protocols are becoming one of the most promising alternatives to the energy-consuming proof-of-work protocols.
However, one particularly critical threat in the PoS setting is the well-known long-range attacks caused by secret key leakage (LRSL attack). Specifically, an adversary can attempt to control/compromise accounts possessing substantial stake at some past moment such that double-spend or erase past transactions, violating the fundamental persistence property of blockchain. Puncturable signatures provide a satisfying solution to construct practical proof-of-stake blockchain resilient to LRSL attack, despite of
the fact that existent constructions are not efficient enough for
practical deployments.

In this paper, we provide an in-depth study of puncturable signatures and explore its applications in the proof-of-stake blockchain. We formalize a security model that allows the adversary for adaptive signing and puncturing queries, and show a construction with efficient puncturing operations based on the Bloom filter data structure and strong Diffie-Hellman assumption. The puncturing functionality we desire is
for a particular part of message, like prefix, instead of the whole
message. Furthermore, we use puncturable signatures to construct practical proof-of-stake blockchain protocols that are resilient to LRSL attack, while previously the forward-secure signature is used to immunize this attack. We implement our scheme and provide experimental results showing that in comparison with the forward-secure signature, our construction performs substantially better on signature size, signing and verification efficiency, significantly on key update efficiency.

\end{abstract}


\begin{IEEEkeywords}
Puncturable signature, Bloom filter, Proof-of-Stake, Blockchain.
\end{IEEEkeywords}}

\maketitle

\section{Introduction}

Proof-of-stake (PoS) protocols have been envisioned as a more ecological alternative of the proof-of-work (PoW) mechanism to maintain consensus, since PoS leverages the virtual resources (e.g., the stake of a node) rather than consuming vast amounts of energy to solve the computational puzzles in PoW. In a proof-of-stake blockchain protocol, roughly speaking, participants randomly elect one party to produce the next block by running a ``leader election" process with probability proportional to their current stake held on blockchain.

In spite of high efficiency, proof-of-stake blockchains only account for a tiny percentage of existing digital currencies market, mainly due to the fact that most existing proof-of-stake protocols suffer from the well-known long-range attacks \cite{li2017securing}\cite{stakebleeding}\cite{deirmentzoglou2019a}(stemming from its weak subjectivity and costless simulation) which degrades security in the blockchain. A oft-cited long-range attack is caused by secret key leakage (abbreviated as LRSL attack in this paper). Specifically, an adversary can attempt to bribe/corrupt the secret keys corresponding to accounts with a substantial stake at some past moment (however, currently low-stake or even no stake), and then construct a fork and alter the history from the point in the past with the majority stake. Moreover, the attack can be sustained due to the fact that the adversary can continue to hold majority stake (e.g., by the reward fees of generating or issuing blocks).
Note that the accounts with low-stake are highly vulnerable to the secret key leakage since they might not be well protected compared with other active accounts, which further aggravates this attack.

Puncturable signature (PS), introduced by Bellare et al.~\cite{bellare2016new}, provides a satisfying solution to construct practical proof-of-stake blockchain protocols resilient to LRSL attack. Generally speaking, a puncturable signature scheme provides a {\sf Puncture} functionality that, given a secret signing key and a particular message $m$, generates an updated secret key which is able to sign all messages except for the punctured message $m$. In this paper, we further generalize the definition of the puncturable signature, particularly, the strings associated with the punctured signing key can be any part of signed messages (e.g. its prefix). In proof-of-stake protocols, the leader $U$ (elected for issuing block) signs the block $B_i$ with the puncturable signature by the secret key $sk_{U}$ at some time slot $sl_i$, where $sl_i$ is the part of block $B_i$, and then $U$ performs puncturing operation on message $sl_i$ which results in an updated $sk'_U$. More specifically, if no empty block is allowed, the punctured message can be $H(B_{i-1})$ instead of $sl_i$, where $H(B_{i-1})$ is the part of block $B_i$ and $B_{i-1}$ is the previous block of $B_i$. The security of puncturable signature guarantees that anyone cannot sign another data block $B'_i$ with the same $sl_i$ (or $H(B_{i-1})$) even though $sk'_U$ is exposed, and thus LRSL attack can be avoided.

A natural way to remedy LRSL attack in proof-of-stake blockchain protocols is to use the forward secure signature~\cite{bellare1999forward}, which guarantees that the previously issued signatures remain valid even if the current secret signing key is compromised. However, the computation performance of forward secure signature depends on either the time periods set in advance or the time periods elapsed so far logarithmically (even linearly), which brings undesirable consumption and becomes
a fatal issue for blockchain applications. Moreover, most signers have no chance to do
any signing within one period but they have to update the signing
key as long as the current period ends, which makes the
update operation a vain effort in the proof-of-stake blockchain. In fact, the forward secure signature can be treated as one special kind of the puncturable signature where the punctured message is the earlier period of time.

Puncturable signatures can also be used in many other scenarios such as asynchronous transaction data signing services. Transaction data signing is a process which guarantees the integrity and authenticity of the sensitive transaction data, such as payment instruction or transaction information of buying a real estate offering. In many cases, using ordinary digital signatures is not enough for these applications, as they often fail to ensure the integrity of past messages in the case when a user's key is compromised. This is particularly challenging in the non-interactive and asynchronous message system, where users may not be online simultaneously and messages may be delayed for substantial periods due to delivery failures and connectivity issues. Similar problem also exists in the theoretical part. For instance, in non-interactive multiparty computation (NI-MPC)~\cite{halevi2017non}, where a group of completely asynchronous parties can evaluate a function (e.g. for the purpose of voting) over their joint inputs by sending a signed message to an evaluator who computes the output. The adversary would control the final output if he can corrupt some parties within a period of time. In these examples, the transaction session ID can be used as a prefix, and after the honest user signs the transaction data (or message), the prefix is punctured so that no other signature exists for messages agreeing on the same prefix. Therefore, the integrity of transaction data (or message) is ensured by puncturable signatures.

\subsection{Our Contributions}
In this work, we provide an in-depth study of the puncturable signature and its applications in the proof-of-stake blockchain protocols. Our overall goal is to design a puncturable signature
that allows for fine-grained revocation of signing
capability with minimum computation cost, and make it a suitable building
block to construct secure and practical proof-of-stake blockchain protocol.
More specifically, our technical contributions are threefold.
\begin{spacing}{1.2}
\end{spacing}
{\it Puncturable signature and its construction.} We introduce the notion of puncturable signature with extended puncturing functionality where the secret key can be updated by puncturing any particular part of message (for simplicity, we use the prefix of message in this paper) instead of puncturing the whole message. In the security model we propose, in addition to making adaptive signing and puncturing queries, adversary also has (one-time) oracle access to a featured {\sf Corruption} oracle, by which the adversary can obtain the current secret key if the challenging string is in the puncturing set $P$. Then we show a construction of puncturable signature based on the probabilistic Bloom filter data structure~\cite{bloom1970space} is secure under our security model. Our PS construction is inspired by an elegant work~\cite{derler2018bloom}, where the authors show how to construct puncturable encryption based on Bloom filter. However, different from the expanded ($k$ times) ciphertext size of underlying encryption scheme in \cite{derler2018bloom}, in our construction, the signature size is almost equal to that of the underlying signature scheme.

In comparison with two prior puncturable signature schemes \cite{bellare2016new}\cite{halevi2017non}, our construction achieves significant efficiency improvement in both signing and puncturing operations. More specifically, the construction in~\cite{bellare2016new} relies on indistinguishability obfuscation, which incurs prohibitive computational burden in practice, while the other one~\cite{halevi2017non} needs to update public key for every puncturing, which has some theoretical merits but hard to implement in real world deployment. On the contrary, in our construction, a puncturing operation only involves a small number of efficient computations (i.e., hashing), plus the deletion of certain parts of the secret key, which outperforms previous schemes by orders of magnitude. Indeed, the puncturable signature is not a simple inverse operation of the puncturable encryption, which is also the reason for no efficient puncturable signature scheme even though efficient puncturable encryption constructions have been proposed for a long time. The crucial difficulty in designing the puncturable signature scheme is how to bind the private key with punctured messages such that the updated private key cannot sign for punctured messages.
\ignore{
{\it Tag-based puncturable signature.} For our puncturable signature scheme based on Bloom filter, the signing algorithm may output $\bot$ for messages whose prefix is not punctured. This is caused by the false positive probability in Bloom filter, and the probability it happens is closely related to the size of secret key and the number of puncturing performed. Put simply, the lower the error probability, the larger the size of secret key and the smaller number of puncturing performed. Therefore, to maintain a balance between space efficiency and error probability, we introduce a new primitive, called tag-based puncturable signatures. In particular, in the lifetime of public key, an ordered tag is updated as long as puncturing operation times reach a pre-set limit, and correspondingly the Bloom filter is reset. We present a generic construction based on  Bloom filter from hierarchical identity based signature (HIBS) scheme, and prove our construction is secure against adaptive puncturing if the underlying HIBS is secure. The intuition behind the construction combines a binary tree approach with our construction of puncturable signatures, where each tag corresponds to a leaf of an ordered binary tree of depth $d$. Our tag-based construction and its security analysis are  independent of any particular instantiation of building blocks, HIBS and PS.}

\begin{figure*}[t]
\centering
\subfigure[Sign time]{\includegraphics[width=0.245\textwidth]{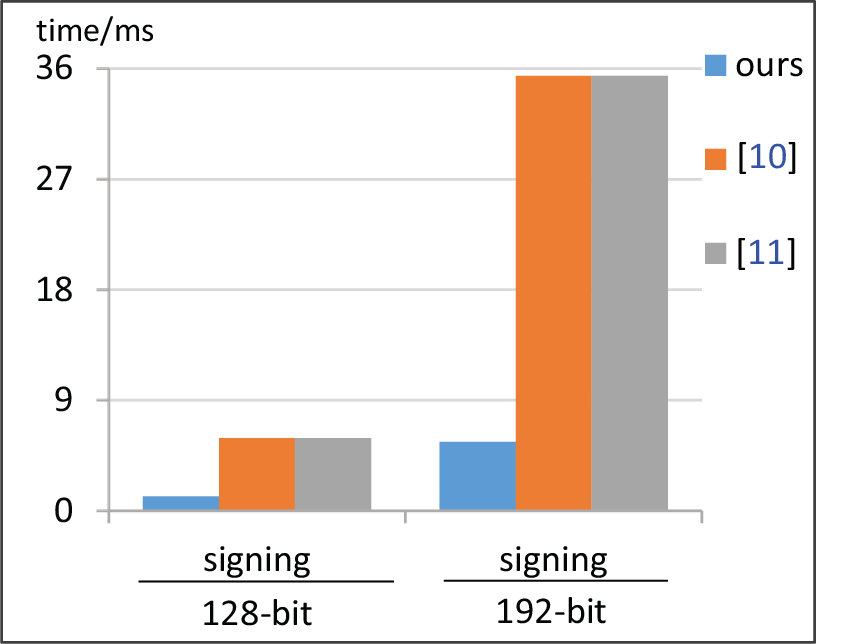}}
\subfigure[Verify time]{\includegraphics[width=0.245\textwidth]{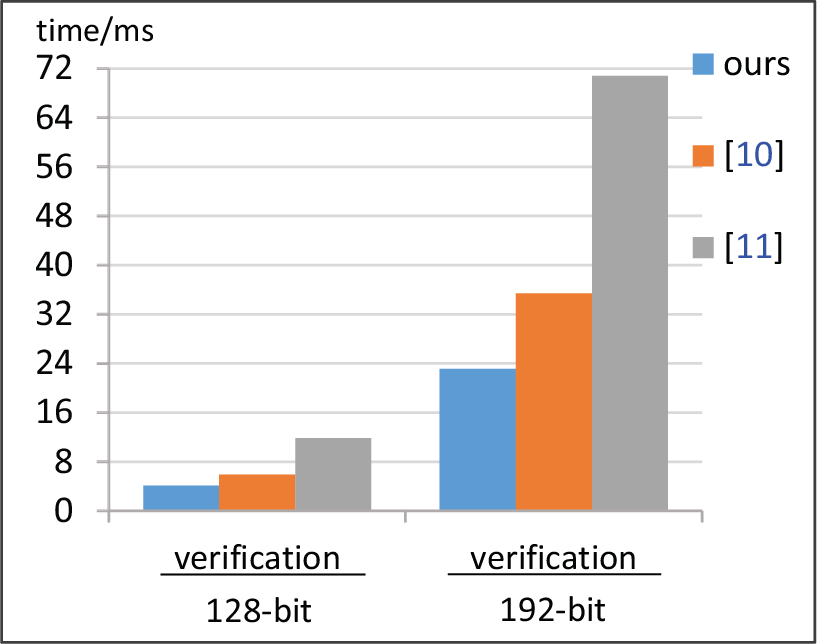}}
\subfigure[Update time]{\includegraphics[width=0.245\textwidth]{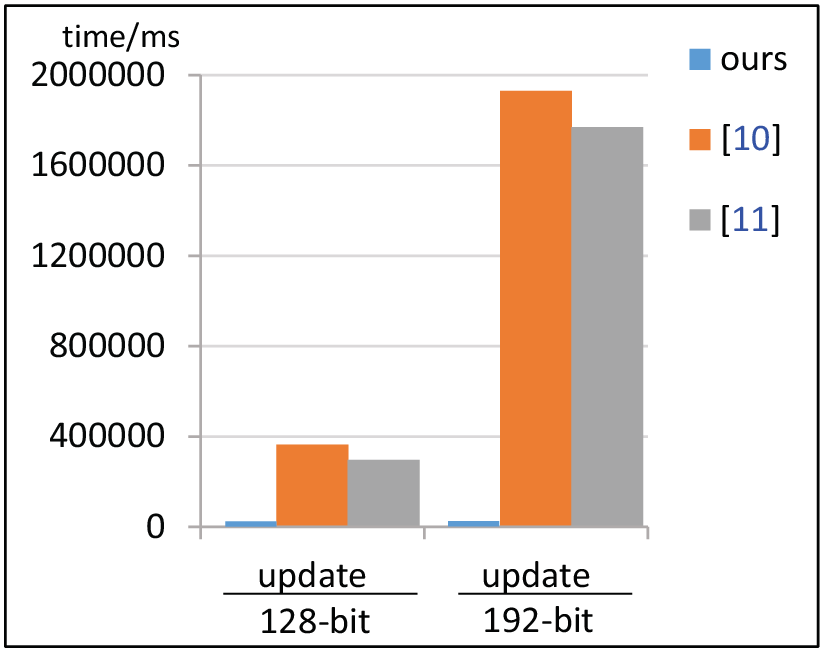}}
\subfigure[Signature size]{\includegraphics[width=0.245\textwidth]{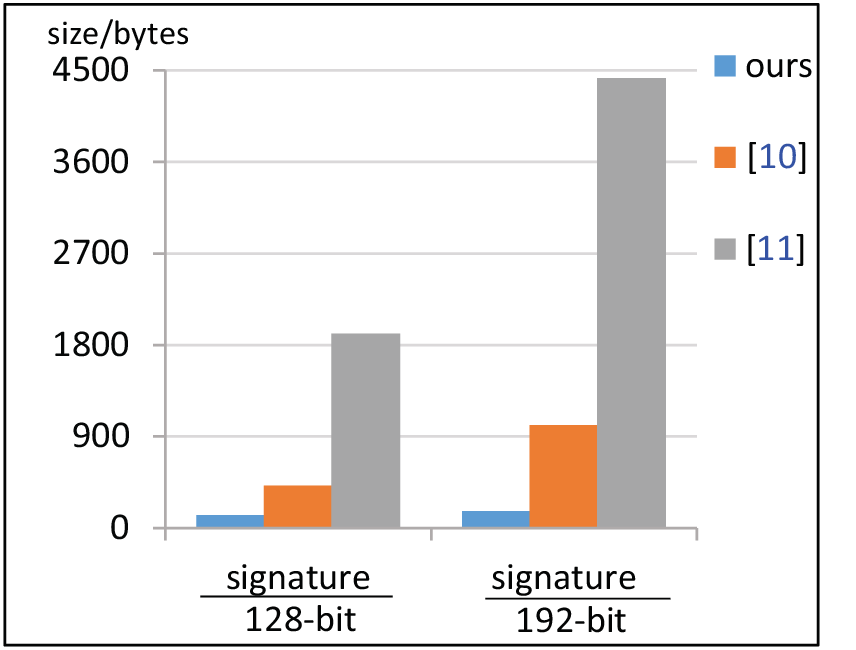}}
\caption*{\bf{Figure 1: Efficiency Comparison }}
\label{fig}
\end{figure*}
\begin{spacing}{1.2}
\end{spacing}
{\it Applications in Proof-of-stake blockchain protocols.} We use the puncturable signature to construct practical proof-of-stake
blockchain protocols that are resilient to LRSL attack. Ouroboros Paros \cite{david2018ouroboros}, a proof-of-stake blockchain protocol, has a real-world implementation in Cardano platform. We present an ideal functionality $\mathcal{F}_{\textsf{PS}}$ of the puncturable signature scheme, and replace $\mathcal{F}_{\textsf{KES}}$ of the forward secure signature scheme in Ouroboros Paros~\cite{david2018ouroboros} protocol with $\mathcal{F}_{\textsf{PS}}$. Then we show that the properties (common prefix, chain quality and chain growth) of Ouroboros Paros protocol remain true in the replaced setting.  Finally, we prove our puncturable signature scheme can securely realize $\mathcal{F}_{\textsf{PS}}$ in the universal composability (UC) security framework. Without loss of generality, we also discuss how to apply our construction to other proof-of-stake blockchain protocols to resist LRSL attack.

{\it Efficiency comparison with forward secure signatures.}
As mentioned above, despite that forward secure signatures can be used to resist LRSL attack, most of the existing schemes have poor performance on key update as well as other operations, often depending on the time period number linearly, which is unsuitable for blockchain application. We conduct experiments evaluating the overhead of deploying our puncturable signature construction and existing forward security signature schemes \cite{itkis2001forward}\cite{malkin2002efficient} at both 128-bit and 192-bit security levels. Figure 1 illustrates the efficiency comparison and the results show that our scheme performs substantially better on signature size, signing and verification efficiency, significantly on key update efficiency, which reduces both communication and computation complexity. In fact, we can replace the ordinary signature with our puncturable signature construction in any other proof-of-stake protocols such as Ouroboros~\cite{kiayias2017ouroboros} and Snow White~\cite{daian2016snow} protocols. Due to the fact that our construction can retain the efficiency of the underlying scheme on signing and verifying, with additional $k$ hash computations, the improved protocols can provide resilience to LRSL attack at almost no additional computing cost.

\subsection{Related Work}

{\it Puncturable signature.} A puncturable signature scheme allows to update its signing key $\sk$ for an arbitrary message $m$ such that the resulting punctured signing key can produce signatures for all messages except for $m$. It is introduced by Bellare et al.~\cite{bellare2016new} as a tool to provide several negative results on differing-inputs obfuscation. However, their construction is based
on indistinguishability obfuscation ~\cite{garg2013candidate} and one-way function, thus, does not yield practical schemes. Moreover, it requires that the punctured signing key is associated with the full signed message. In contrast our construction is based on the $\tau$-SDH assumption, and the string associated with the punctured signing key can be any particular part of the message to be signed (e.g. the prefix of the message), which is more flexible and applicable. Halevi et al.~\cite{halevi2017non} also propose a puncturable signature scheme which is puncturable at any prefix of the signed message. However, their puncturable operation needs to update public keys repeatedly. In practice, it is  inefficient to verify the updated public keys continuously and it is also difficult to let each user in the system maintain other users' public keys updated.
\begin{spacing}{1.1}
\end{spacing}
{\it Delegation related signature.} Policy-based signature, introduced by Bellare et al.~\cite{PKC:BelFuc14}, allows a signer to only sign a message satisfying some policy specified by an authority. Another related primitive, called functional signature, is introduced in the work~\cite{PKC:BoyGolIva14}. In functional signature, only the master signing key can be used to sign any message, while the authority provides the secondary signing keys for functions $f$ that allow the signing of any message in the range of $f$. Delegatable functional signature is introduced in the work~\cite{PKC:BacMeiSch16} and supports the delegation of signing rights with a policy to the designated third party. Append-only signatures (AOS) \cite{kiltz2005append-only} is also a related primitive, in particular, given an AOS signature on a message $m$, one can only compute an AOS signature on any message $m'$ with $m$ as the prefix. Different from above primitives, puncturable signature provides a puncture functionality that may repeatedly update the secret key to {\it revoke} the signing capability for selected messages in a dynamic manning.

\begin{spacing}{1.1}
\end{spacing}
{\it Forward secure signature.} A forward secure signature scheme guarantees the adversary with the compromised secret key cannot forge signatures relative to previous time periods. It is proposed by Anderson~\cite{anderson1997two} and then formalized by Bellare et al.~\cite{bellare1999forward}. Prior forward secure signatures are constructed by either using arbitrary signature schemes in a black-box manner~\cite{malkin2002efficient}\cite{krawczyk2000simple} or modifying the existing signature schemes~\cite{bellare1999forward}\cite{itkis2001forward}\cite{abdalla2000new}. All these forward secure schemes except for~\cite{malkin2002efficient}, the number of time periods $T$ (arbitrarily large) must be set in advance, such that the performance depends on $T$ logarithmically or even linearly. Nevertheless, the performance in~\cite{malkin2002efficient} still depends on the time periods elapsed so far.
\begin{spacing}{1.1}
\end{spacing}
{\it Proof-of-stake blockchain protocols.} Proof-of-stake protocols were first initiated in online forums and subsequently a number of proof-of-stake protocols were proposed and implemented by the academic community. In order to provide forward security (and also achieve resilience against LRSL attack and other long-range attacks), Ouroboros Paros~\cite{david2018ouroboros} and Ouroboros Genesis~\cite{badertscher2018ouroboros} formalize and realize in the universal composition setting a forward secure digital signature scheme, Algorand~\cite{gilad2017algorand} considers it as one of future work and implements ephemeral key pairs in its updated full version~\cite{chen2017algorand}, whereas Snow White~\cite{daian2016snow} and Ouroboros~\cite{kiayias2017ouroboros} adopt a weak adaptive corruption model and cannot avoid LRSL attack. In addition, several countermeasures have been proposed, such as the trusted execution environments~\cite{li2017securing}, checkpointing mechanism~\cite{nxt}, and punishment mechanism revealing the real identity~\cite{li2017securing}, revealing the signing key~\cite{derler2018generic} or slashing the stake \cite{slashing} of the malicious stakeholder.

\section{Preliminaries}
{\it Notation.} Let $\secparam$ denote the security parameter, $\lfloor x \rfloor$ denote the greatest integer less than or equal to $x$, $[n]$ denote the set of the first $n$ positive integers, and $\ppt$ denote probabilistic polynomial time.
For an array $T \in \{0, 1\}^n$, we let $T[i]$ denote the $i$-th bit of the array, if $i \leq n$.

We say a function $\negl(\cdot): \N \rightarrow (0,1)$ is negligible, if for every constant $c \in \N$, $\negl(n) < n^{-c}$ for sufficiently large $n$.

\subsection{Bloom Filter}
A Bloom filter~\cite{bloom1970space} is a probabilistic data structure designed to solve the approximate set membership problem. Specifically, it can be used to test whether an element is a member of the elements set $\mathcal{S}$ from a large universe $\mathcal{U}$ rapidly and space-efficiently, where the false-positive error is allowed. Put simply, for any $s \in \mathcal{S}$
the Bloom filter always outputs 1 (``yes''), and for any $s \notin \mathcal{S}$ it outputs 1 only with some  small probability.

\begin{definition}[Bloom Filter\cite{derler2018bloom}]\label{defn:bloomf}
A Bloom filter $\BF$ for set $\mathcal{U}$ consists of a tuple of PPT algorithms $(\gen, \update, \check)$.
\end{definition}
\begin{itemize}[leftmargin=*]
 \item $\gen(\ell, k)$: On input two integers $\ell, k \in \mathbb{N}$, the algorithm first samples $k$ universal hash functions $H_1, \ldots, H_k$ which generate a uniform distribution over the range $\{1,...,l\}$, then initializes the $\ell$-bit array $T = 0^\ell$ and sets $H = \{H_i\}_{i\in [k]}$, and finally returns $(H, T)$.
 \item $\update(H, T, u)$: To add an element $u \in \mathcal{U}$ to $\mathcal{S}$, the algorithm updates $T[H_i(u)] = 1$ for all $i\in [k]$, leaves other bits of $T$ unchanged, and returns $T$.
 \item $\check(H, T, u)$: To check whether an element $u$ is in $\mathcal{S}$, the algorithm returns 1 (``yes") if $\wedge_{i \in [k]} T[H_i(u)] = 1$, otherwise returns 0 (``no").
\end{itemize}

\noindent{\it Properties of Bloom filter.} The main properties of Bloom filter used in this paper can be summarized as follows:
\begin{description}[leftmargin=*]
\item \textbf{Correctness}: A Bloom filter always outputs 1 for elements that have already been added to the set $\mathcal{S}$. More precisely, if $s \in \mathcal{U}$ has been added to $\mathcal{S}$ and the array $T$ has been updated by running $\update(H, T, s)$, then we have $\prob[\check(H, T, s)=1] = 1.$

\begin{spacing}{0.3}
\end{spacing}

\item \textbf{False-positive errors}: A Bloom filter may yield a false-positive error, where an element $s^*$ is erroneously judged to be in $\mathcal{S}$ by the $\check$ algorithm even though $s^*$ has not yet been added to the set $\mathcal{S}$. More precisely, the error probability is bounded by $\prob[\check(H, T, s^*) = 1] \approx (1 - e^{k n / \ell})^k$ for any $s^* \notin \mathcal{S}$.

  \begin{spacing}{0.3}
\end{spacing}

\item \textbf{Space advantages}:  The size of representation $T$ is a constant number (i.e., $l$ bits) independent of the size of $\mathcal{S}$. However, as the size of $\mathcal{S}$ increases, false-positive probability of Bloom filter would increase.

\end{description}
\begin{spacing}{1.2}
\end{spacing}
{\it Discussion on the choice of parameters.} In bloom filter, assuming the optimal number of hash function $k$ to achieve the smallest false-positive probability $pr$, we obtain a size of the bloom filter given by $\ell = -\frac{n\ln pr}{(\ln 2)^2}$, and the optimal $k$ is given by $k = \lceil \frac{\ell}{n}\ln 2 \rceil$. For example, when $pr = 10^{-3}$ and $n = 2^{20}$, we have $\ell \approx 2$ MB and $k = 10$.

\ignore{Looking ahead the puncturable signature in Section 3, one selected curve would yield a signature of size about 2$|\mathbb{G}_1|$ bits. The public key and secret key are of size about $|\mathbb{G}_1|$ bits and $|m| \cdot |\mathbb{G}_1| + |m|$ bits respectively, where $|a|$ denotes the binary length of a string $a$. Note that, the secret key size may be very large as discussed in~\cite{derler2018bloom}, however, we need to emphasize that the secret key generation can be done off-line, and the initial secret key has its maximum size, since every puncture (i.e. corresponding to one signature operation) reduces that size. Moreover, the key update (puncture) process is very efficient and just takes a few ($k$) hash computations and several deletion operations. Finally, through the HIBS-based puncturable signature in Section~\ref{sec:hibs}, we can further largely reduce the secret size by reducing the maximum number of puncturing at the cost of frequent key delegations.}

\subsection{Bilinear Groups and $\tau$-SDH Assumption}
We say that $\mathcal{G}$ is a bilinear group generator if given the security parameter $\secparam$, it outputs a tuple $\params = (p, e, \psi, \G_1, \G_2, \G_T,$ $ P_1, P_2)$, where $\G_1$, $\G_2$ and $\G_T$ are three groups having prime order $p$, $P_i$ is the generator of $\G_i$ for $i \in \{1,2\}$, and  $e: \G_1 \times \G_2 \rightarrow \G_T$ is a non-degenerate map satisfying:
\begin{itemize}[leftmargin=*]
 \item Bilinearity: For any $(P, Q) \in \G_1 \times \G_2$ and $a, b \in \Z^*_p$, $e(aP, bQ) = e(P, Q)^{ab}$ holds.
 \item Non-degeneracy: $e(P, Q) = 1$ for any $Q \in \G_2$ iff $P = \mathcal{O}$.
 \item Computability: There exists an efficient algorithm to compute $e(P, Q)$ for any $(P,Q) \in \G_1 \times \G_2$.
 \item  There exists an publicly and efficiently computable isomorphism $\psi: \G_2\rightarrow\G_1$ such that $\psi(P_2)=P_1$.
\end{itemize}

The security of our scheme is based on $\tau$-strong Diffie-Hellman ($\tau$-SDH) assumption, which was previously formalized in \cite{boneh2004short} and \cite{mitsunari2002a}.

\begin{definition}[$\tau$-Strong Diffie-Hellman Assumption ($\tau$-SDH)]\label{defn:dh}
 Let $\params = (p, e, \psi, \G_1, \G_2, \G_T, P_1, P_2) \leftarrow \mathcal{G}(1^\lambda)$ and $(P_1, P_2, \alpha{P_2}, \alpha^2{P_2}, ... ,\alpha^{\tau}{P_2})$ be a $\tau + 2$ tuple for $\alpha \in \Z^*_p$.
We say $\tau$-SDH assumption holds if for any $\ppt$ adversary $\A$, $\prob[(h,\frac{1}{h+\alpha}P_1) \leftarrow \A(P_1, P_2, \alpha{P_2}, \alpha^2{P_2}, ... ,\alpha^{\tau}{P_2})] \leq \negl(\secparam)$ with $h \in \Z^*_p$.
\end{definition}

\section{Puncturable Signatures}
In this section, we formalize the syntax and security definition of puncturable signature, and then we propose a puncturable signature scheme and prove its security under the $\tau$-SDH assumption.

\subsection{Syntax and Security Definition}
\ignore{Let space $M$ be $\{0, 1\}^\secparam$, and the message space is $\mathcal{M} = M^*$. A puncturable signature scheme $\Sigma$ consists of a tuple of $\ppt$ algorithms $\Sigma$ = $($$\setup,$ $\punc,$ $\sign,$ $\verify$$)$ with descriptions as follows:}
Let the message space be $\mathcal{M}$. A puncturable signature scheme $\Sigma$ consists of a tuple of $\ppt$ algorithms $\Sigma = (\setup, \punc,$ $\sign, \verify)$ with descriptions as follows:
\begin{itemize}[leftmargin=*]
 \item $\setup(1^\secparam, \ell , k)$: On input the security parameter $\secparam$, parameters $\ell$ and $k$ for the Bloom filter, the setup algorithm outputs the public key $\vk$ as well as the secret key $\sk$.
 \item $\punc(\sk, str)$: On input the secret key $\sk$ and a string $str \in \mathcal{M}$, the puncturing algorithm outputs the updated secret key $\sk'$ which is also called the punctured secret key. We also say that $str$ has been punctured.
 \item $\sign(\sk, m)$: On input the secret key $\sk$ and a message $m$, it outputs a signature $\sigma$.
 \item $\verify(\vk, m, \sigma)$: On input the public key $\vk$, a signature $\sigma$ and message $m$, the verification algorithm outputs 1 if $\sigma$ is a valid signature for $m$. Otherwise, it outputs 0.
\end{itemize}

{\it Correctness of puncturable signatures.} Intuitively, the correctness requires that (1) signing is always successful with the initial, non-punctured secret key, (2) signing fails when attempting to sign a message with a prefix that has been punctured, and (3) the probability that signing fails is bounded by some non-negligible function, if the prefix of the message to be signed has not been punctured.

\begin{definition}[Correctness] \label{defn:puncor}
For any message $m$ with prefix $m'$, any $(\sk_{\sf init}, \vk) \leftarrow \setup(1^\secparam)$, and any sequence of invocations of $\sk \leftarrow \punc(\sk, .)$, we have
\begin{enumerate}[leftmargin=*]
\item $\verify(\vk, m, \sign(\sk_{\sf init}, m)) = 1$, where $\sk_{\sf init}$ is the initial, non-punctured secret key.
\item If $m'$ has been punctured, then $\verify(\vk, m ,\sign(\sk', m)) = 0$, where $\sk' \leftarrow \punc(\sk, m')$.
\item Otherwise, it holds that $\Pr[\verify(\vk, m, \sign(\sk, m))$ $\neq 1] \leq$ $\mu(\ell, k)$, where $\mu(\cdot)$ is a (possibly non-negligible) bound.
\end{enumerate}
\end{definition}

\begin{remark}
\ignore{\emph{We note that the puncturing functionality defined above is only for message-prefix. Since the message space $\mathcal{M}$ is concatenation of equal length bit strings, we can extend the puncturing functionality by augment the puncturing string by an index $i$, i.e. $(i, m)$, so the signing algorithm fails for message $m'$, if there exists $(i, m) \in P$, such that $m'|_{i\secparam}^{(i + 1)\secparam} = m$. For simplicity, we still use prefix-puncturing through this paper, but both the definitions and our constructions can be easily extended to support the extended puncturing functionality.}}

\emph{We notice that the puncturing functionality defined above is for message-prefix, whose length can be determined in specific implementation (e.g., the slot parameter in proof-of-stake blockchain). In specific applications, the message $m$ to be signed can be split into $n$ ($n \geq 1$) parts denoted by $m = m_1||...||m_i||...||m_n$, where different parts may have different lengths and different semantics, for example, $m_1$ denotes the time stamp and the remaining denotes the message specifics. We can extend the puncturing functionality by
puncturing strings at arbitrarily pre-defined position (even the whole message), e.g. $i$-th part, which means the signing algorithm fails for message $m$ = $m_1\|...\|m_i\|...\|m_n$ if $m_i$ has been punctured. For simplicity, we still use prefix-puncturing in this paper, but both the definitions and our constructions can be easily extended to support the general puncturing functionality.}

 \ignore{{\color{red}and in the application to PoS blockchain protocols, the puncturing string is the slot parameter at the postfix of messages.}}

\end{remark}

{\it Security Definition.} Intuitively, the security definition of puncturable signature scheme requires that the adversary cannot forge signatures on messages having been punctured even though the punctured secret key is compromised. Formally, for any $\ppt$ adversary $\A$, we consider the experiment $\expt^{\sf ps}_\A(1^\secparam)$ between $\A$ and the challenger $\mathcal{C}$:
\begin{enumerate}[leftmargin=*]
 \item \textbf{Setup}: $\mathcal{C}$ computes $(\vk, \sk) \leftarrow \setup(1^\secparam)$ and sends $\vk$ to the adversary $\A$. 
     $\mathcal{C}$ initializes two empty sets $Q_{\sf sig} = \emptyset$ and $P = \emptyset$.

 \item \textbf{Query Phase}: Proceeding adaptively, the adversary $\A$ can submit the following two kinds of queries:
 \begin{itemize}[leftmargin=*]
  \item \textbf{Signature query}: On input a message $m$ from the adversary $\A$, $\mathcal{C}$ computes $\sigma \leftarrow \sign(\sk, m)$ and updates $Q_{\sf sig} = Q_{\sf sig} \cup \{m\}$. Then $\mathcal{C}$ sends back $\sigma$.
  \item \textbf{Puncture query}: On input a string $str$, $\mathcal{C}$ updates $\sk$ by running $\punc(\sk, str)$, and updates $P$=$P$ $\cup$ $\{str\}$.
 \end{itemize}
\item \textbf{Challenge Phase}: $\A$ sends the challenging string $m'$ to $\mathcal{C}$, and $\A$ can still submit signature and puncture queries as described in the \textbf{Query phase}.

\item \textbf{Corruption query}: The challenger returns $\sk$ if $m' \in P$ and $\bot$ otherwise.

\item \textbf{Forgery}: $\A$ outputs a forgery pair $(m, \sigma)$.
\end{enumerate}
We say that adversary $\A$ wins the experiment $\expt^{\sf ps}_\A(1^\secparam)$ if $m'$ is the prefix of $m$ and
$m \notin Q_{\sf sig} \wedge \verify(\vk, m, \sigma) = 1.$

\ignore{
\begin{itemize}
 \item $m = m^* \cdots \wedge m \notin Q_{\sf sig} \wedge \verify(\vk, m, \sigma) = 1$.
 \item $m \neq m^* \cdots \wedge m \notin Q_{\sf sig}$, and there exists $m' \in Q_{\sf punc}$, such that $m = m' \cdots$, and $ \verify(\vk, m, \sigma) = 1$.
\end{itemize}
}

\begin{definition}[Unforgeability with adaptive puncturing] \label{defn:unforg}
We say the puncturable signature scheme $\sf ps$ is unforgeable with adaptive puncturing, if for any $\ppt$ adversary $\A$, the probability of winning experiment $\expt^{\sf ps}_\A(1^\secparam)$ is $\negl(\secparam)$, where the probability is over the randomness of the challenger and adversary.
\end{definition}

\subsection{Our Construction}
We present a puncturable signature construction based on the Chinese IBS, an identity-based signature scheme standardized in ISO/IEC 14888-3~\cite{iso14883-3}.

The key idea of our construction is to derive secret keys for all Bloom filter bits $i \in [l]$ using IBS schemes, and then bind the prefix string $m'$ with $k$ positions where the secret keys are used to sign messages with prefix $m'$. In addition, puncturing at $m'$ implies the deletion of keys in the corresponding positions.

Let $(p, e, \psi, \G_1, \G_2, \G_T, P_1, P_2) \leftarrow \mathcal{G}(1^\lambda)$, and $\BF= (\BF.\gen,$ $\BF.\update, \BF.\check)$ be a Bloom filter. Choose a random generator $P_2 \in \G_2$, and set $P_1=\psi(P_2) \in \G_1$. Let $h_1: \mathbb{N}\rightarrow \mathbb{Z}^*_p$ and $h_2 : \{0, 1\}^* \times \mathbb{G}_T \rightarrow \mathbb{Z}^*_p$ be cryptographic hash functions, which we model as random oracles in the security proof. The public parameters are ${\sf params}:= (p, e, \psi, \G_1, \G_2, \G_T, P_1, P_2, h_1, h_2)$ and all the algorithms described below implicitly take $\sf params$ as input. The construction of the puncturable signature scheme $\sf ps = (\setup, \punc,$ $\sign, \verify)$ is the following:
\begin{itemize}[leftmargin=*]
 \item $\setup(1^\secparam, \ell, k)$: This algorithm first generates a Bloom filter instance by running $(\{H_j\}_{j\in [k]}, T) \leftarrow \BF.\gen(\ell, k)$. Then it chooses $s \rsample\mathbb{Z}^*_p$ and outputs
      $$\sk= (T, \{sk_i\}_{i\in [\ell]}),\quad  \vk=(P_{pub}, g, \{H_j\}_{j\in [k]}),$$
where $sk_i=\frac{s}{s+h_1(i)}P_1$, $P_{pub}=sP_2$, and $g = e(P_1, P_{pub})$.

 \item $\punc(\sk, str)$: Given a secret key $\sk = (T, \{\sk_i\}_{i\in [\ell]})$ where $\sk_i = \frac{s}{s+h_1(i)}P_1$ or $\bot$, this algorithm first computes $T'$ = $\BF.\update$ ($\{H_j\}_{j\in [k]}$, $T$, $str$). Then, for each $i\in [\ell]$ define
     $$
\sk'_i = \left\{
             \begin{array}{ll}
             \sk_i, & \text{if} \ T'[i] = 0\\
             \bot, & \text{otherwise}\\
             \end{array}
\right.
$$
where $T'[i]$ denotes the $i$-th bit of $T'$. Finally, the algorithm returns $\sk' = (T', \{\sk'_i\}_{i\in [\ell]})$.

\item $\sign(\sk, m)$: Given a secret key $\sk = (T, \{\sk_i\}_{i\in [\ell]})$ and a message $m$ with the prefix $m'$, the algorithm first checks whether $\BF$.$\check$ ($\{H_j\}$$_{j\in [k]}$, $T$, $m'$$) = 1$ and outputs $\bot$ in this case. Otherwise, note that $\BF$.$\check$($\{H_j\}_{j\in [k]}$, $T$, $m'$) = 0 means there exists at least one index $i_{j} \in \{i_1,\cdots,i_k\}$ such that $\sk_{i_{j}}\neq \bot$, where $i_j = H_j(m')$ for $j\in [k]$. Then it picks a random index $i_{j^*}$ such that
    $\sk_{i_{j^*}} = \frac{s}{s+h_1(i_{j^*})}P_1$. Choose $x \rsample \mathbb{Z}^*_p$ and compute $r = g^x$, then set
    $$h=h_2(m,r),S=(x - h)\sk_{i_{j^*}}.$$
    The output signature on $m$ is $\sigma = (h,S,i_{j^*})$.

 \item $\verify(\vk, m, \sigma)$: Given a public key $\vk = (P_{pub}, g, \{H_j\}_{j\in [k]})$, a message $m$ with the prefix $m'$, and a signature $\sigma = (h,S,i_{j^*})$, the algorithm checks whether
     $$i_{j^*} \in S_{m'} \wedge h = h_2\big(m,e(S,h_1(i_{j^*})P_2+P_{pub})g^{h}\big),$$
     where $S_{m'} = \{H_j(m'): j \in [k]\}$. If it holds, the algorithm outputs 1, and 0 otherwise.

\end{itemize}

\begin{lemma}\label{defn:lemma correctness1}
 Our basic construction described above satisfies correctness (c.f. Definition~\ref{defn:puncor}).
\end{lemma}
\begin{proof}
If the secret key is initial and non-punctured, we have
\begin{align*}
&e(S,h_1(i_{j^*})P_2+P_{pub})g^h \\
 &= e((x-h).\frac{s}{h_1(i_{j^*})+s}P_1,(h_1(i_{j^*})+s)P_2).e(P_1,P_2)^{hs}\\
  &=  e(P_1,P_2)^{xs-hs}e(P_1,P_2)^{hs} = g^x = r
  \end{align*}
and then $h=h_2\big(m,r)=h_2(m,e(S,h_1(i_{j^*})P_2+P_{pub})g^h\big)$. Therefore, the first requirement of Definition~\ref{defn:puncor} holds. If $m'$ is punctured, by the correctness of Bloom filter, we have $\BF.\check(H, T, m')$ = 1, which means all the secret keys used to sign messages with $m'$ as the prefix have been deleted. Therefore, the signing of the message $m$ with the prefix $m'$ fails and the second requirement of Definition~\ref{defn:puncor} holds. If $m'$ is not punctured, the correctness error of our construction occurs only when $\BF.\check(H, T, m') = 1$, which is essentially identical to the false-positive probability of the Bloom filter, and the third requirement of Definition~\ref{defn:puncor} holds.
\end{proof}

\begin{remark}\label{rem:bound}
\emph{
In this section, we assume that the false-positive probability from Bloom filter is acceptable, which means the number of puncturings supported by our construction is below a pre-set parameter, depending on the application scenarios and the upper bound $n$ of the Bloom filter. In the security proof below, we assume that the number of puncturing queries is also bounded by the pre-set parameter. }
\end{remark}

\begin{theorem}\label{thm:ps}
Assuming that an algorithm $\mathcal{A}$ wins in the $\expt^{\sf ps}_\A(1^\secparam)$ experiment (c.f. Definition~\ref{defn:unforg}) to our construction ${\sf ps}$, with the advantage $\epsilon_0 \geq 10k(q_S+1)(q_S+q_{h_2})/\big(p(1 - (1 - 1/l)^k)\big)$ within running time $t_0$, $\tau$-SDH assumption can be broken for $\tau = q_{h_1}$ within running time $t_2 \leq 120686q_{h_2}t_0/({\epsilon_0}(1-\tau/p))$, where $q_{h_1}$, $q_{h_2}$ and $q_S$ are the maximum query times of the hash function $h_1$, $h_2$ and signing, respectively.
\end{theorem}

%
\begin{proof}\renewcommand{\qedsymbol}{$\blacksquare$}
In order to prove the security of our scheme, we consider a particular adversary $\mathcal{B}$ with a fixed position against our signature scheme in a variant of the above experiment $\expt^{\sf ps}_\A(1^\secparam)$, denoted by $\expt^{\sf fps}_\A(1^\secparam)$. Specifically, in the \textbf{Setup}, the challenger returns system parameters together with a fixed position $i^* \in [l]$; the following \textbf{Query Phase} remains unchanged and the \textbf{Challenge Phase} can be omitted; in the \textbf{Corruption} query, the challenger only returns the current secret key by excluding the key at position $i^*$. We say $\mathcal{B}$ wins the experiment $\expt^{\sf fps}_\A(1^\secparam)$, if $\mathcal{B}$ outputs the message $m$ together with a valid signature ($h,S,i^*$) on $m$.

We sketch the proof in two steps as in \cite{barreto2005efficient}\cite{PKC:ChoChe03}. First, we prove there exists an algorithm $\mathcal{B}$ that wins in the $\expt^{\sf fps}_\A(1^\secparam)$ experiment with a non-negligible advantage, if the adversary $\mathcal{A}$ has a non-negligible advantage against our signature scheme ${\sf ps}$ in the $\expt^{\sf ps}_\A(1^\secparam)$ experiment (c.f. Lemma~\ref{defn:lemma 1}). Then, assuming the existence of $\mathcal{B}$, we can construct an algorithm $\mathcal{C}$ that breaks the $\tau$-SDH assumption (c.f. Lemma~\ref{defn:lemma 2}).

\begin{lemma}\label{defn:lemma 1}
Assuming that an algorithm $\mathcal{A}$ wins in the $\expt^{\sf ps}_\A(1^\secparam)$ experiment (c.f. Definition~\ref{defn:unforg}) to our construction ${\sf ps}$, with a probability $\epsilon_0$ within running time $t_0$, there exists an algorithm $\mathcal{B}$ that wins in the $\expt^{\sf fps}_\A(1^\secparam)$ experiment as described above to ${\sf ps}$ which has a probability $\epsilon_1 \geq  \epsilon_0(1 - (1 - 1/l)^k)/k$ within running time $t_1 \leq t_0$.
\end{lemma}

\begin{proof}
Suppose there exists such an adversary $\mathcal{A}$, and we construct a simulator $\mathcal{B}$ that simulates an attack environment and uses $\mathcal{A}$'s forgery to win in its own $\expt^{\sf fps}_\A(1^\secparam)$ experiment. The simulator $\mathcal{B}$ can be described as follows:

 \begin{itemize}[leftmargin=*]
\item{\bf Invocation.} $\mathcal{B}$ is invoked on a given position $i^* \in [l]$.

\item{\bf Queries.} $\B$ answers adaptive queries from $\A$ as follows:

 \begin{itemize}[leftmargin=*]
 \item $\mathcal{B}$ makes the $\sf Setup$ query and forwards all the returned parameters to $\mathcal{A}$ for $\mathcal{A}$'s {\sf Setup} query.

\item Before $\mathcal{A}$ outputs the challenging string, $\mathcal{B}$ just forwards the queries of $\mathcal{A}$, including $\sign$, $\punc$, $h_1$ and $h_2$, to its experiment and returns the result to $\mathcal{A}$.

\item When $\mathcal{A}$ outputs the challenging string denoted by $m'$ after the series of queries, $\mathcal{B}$ checks whether $i^*$ $\in$ \{$H_j(m'): j \in [k]$\} and aborts if this does not hold. Otherwise, $\mathcal{B}$ provides the simulation for $\mathcal{A}$ as follows. For the queries $h_1$, $h_2$, $\punc$ and $\sign$, $\mathcal{B}$ just passes them to its challenger and returns the result as before. While for $\sf Corruption$ query, $\mathcal{B}$ firstly checks whether $m' \in P$ and returns $\emptyset$ if this does not hold. Otherwise, $\mathcal{B}$ makes $\sf Corruption$ query in its experiment, and returns the response $\sk$ to $\mathcal{A}$. 
 \end{itemize}
 \end{itemize}
\ignore{If the simulation of $\mathcal{B}$ does not abort, it is easy to verified that $\mathcal{B}$ provides a perfect simulation for $\mathcal{A}$ according to the above discussion.}

Eventually, $\mathcal{A}$ outputs a valid signature $(m, \sigma = (h,S,j^*))$, where $m'$ is the prefix of $m$. If $j^* = i^*$, then $\mathcal{B}$ sets $(m, \sigma)$ as its own output and apparently $\mathcal{B}$ also wins in its $\expt^{\sf fps}_\A(1^\secparam)$ experiment.

In the simulation described above, there are two events that cause $\B$ to abort: (1) $i^* \notin \{H_j(m'): j \in [k]\}$ for the challenging string $m'$; (2) $i^* \neq j^*$ for the forged signature $\sigma = (h,S,j^*)$.

\ignore{Assume that the hash functions in Bloom filter selects each position in array with equal probability. The probability that a certain position is not set to 1 by any of the $k$ hash functions is $(1 - 1 / \ell)^k$. Therefore, the probability that $\B$ does not abort in challenge phase is $P_1 = 1 - (1 - 1 / \ell)^k$.}

Recall that the $k$ hash functions in Bloom filter are sampled universally and independently, and thus each position in the array is selected with an equal probability. Besides, $i^*$ is invisible and looks random to $\mathcal{A}$, then the selection of $m'$ is independent of $i^*$. Therefore the probability that $i^* \notin \{H_j(m'): j \in [k]\}$ is $(1 - 1 / \ell)^k$. Similarly, the second event $i^* \neq j^*$  happens with probability  $1-1 / k$. Combing these, with probability $\epsilon_1 \geq  \epsilon_0(1 - (1 - 1/l)^k)/k$, $\B$ completes the whole simulation without aborting and wins in the $\expt^{\sf fps}_\A(1^\secparam)$ experiment.
\end{proof}
\begin{lemma}\label{defn:lemma 2}
Assuming that an algorithm $\mathcal{B}$ wins in the $\expt^{\sf fps}_\A(1^\secparam)$ experiment to our construction ${\sf ps}$, with an advantage $\epsilon_1 \geq 10(q_S+1)(q_S+q_{h_2})/p$ within running time $t_1$, there exists an algorithm $\mathcal{C}$ that breaks the $\tau$-SDH assumption for $\tau = q_{h_1}$ within running time $t_2 \leq 120686q_{h_2}t_1/({\epsilon_1}(1-{\tau}/{p}))$, where $q_{h_1}$, $q_{h_2}$ and $q_S$ are the maximum query times of the hash function $h_1$, $h_2$ and signing, respectively.
\end{lemma}
\begin{proof}
Suppose there exists such an adversary $\B$, and we construct a simulator $\mathcal{C}$ that simulates an attack environment and uses $\mathcal{B}$'s forgery to break the $\tau$-SDH assumption. The simulator $\mathcal{C}$ can be described as follows:

\begin{itemize}[leftmargin=*]
 \item {\bf Invocation.} Given a random instance $($$P_1$, $P_2$, $\alpha{P_2}$, $\alpha^2{P_2}$,...,$\alpha^{\tau}{P_2}$$)$ of the $\tau$-SDH assumption as input, $\mathcal{C}$ aims to find a pair ($h,\frac{1}{h+\alpha}P_1$) for some $h\in Z^*_p$.

 \item {\bf Setup.} $\mathcal{C}$ randomly samples $\tau-1$ elements $w_1$,$w_2$,$...$,$w_{\tau-1}$ $\xleftarrow{R}$ $Z^*_p$, where $w_i$ ($i \in [\tau -1]$) will be used as the responses to $\B$'s $h_1$ queries. $\mathcal{C}$ expands the polynomial $f(y) = \prod_{i=1}^{\tau-1}(y+w_i)$ to obtain the the coefficients $c_0,c_1,...,c_{\tau-1} \in Z^*_p$ such that $f(y) = \sum_{i=0}^{\tau-1}c_iy^i$ .

  \ \ \  Then $\mathcal{C}$ sets $\hat{P_2}=\sum_{i=0}^{\tau-1}c_i\alpha^iP_2 = f(\alpha)P_2 \in \G_2$ and $\hat{P_1} = \psi (\hat{P_2}) \in \G_1$ to be new generators of $\G_2$ and $\G_1$, respectively. The master public key is set to $\hat{P}_{pub} = \sum_{i=1}^{\tau}c_{i-1}(\alpha^iP_2)$ such that $\hat{P}_{pub}=\alpha \hat{P_2}$. Note that the master secret key is the unknown $\alpha$.

  \ \ \ To provide the secret keys corresponding to positions having been queried to $h_1$, $\mathcal{C}$ expands $f_i(y)=f(y)/(y+w_i)=\sum_{i=0}^{\tau-2}d_iy^i$ and
$$
            \sum_{i=0}^{\tau-2}d_i\psi(\alpha^{i+1}P_2) = \alpha f_i(\alpha)P_1 = \frac{\alpha f(\alpha)P_1}{\alpha+w_i} = \frac{\alpha}{\alpha+w_i}\hat{P_1}\ \ (1)
$$

 Therefore, the $\tau-1$ pairs ($w_i,\frac{\alpha}{\alpha + w_i}\hat{P_1}$) for $i \in [\tau -1]$ can be computed using the left member of equation (1).

  \ \ \ Then $\mathcal{C}$ provides the parameters $(p$, $e$, $\psi$, $\G_1$, $\G_2$, $\G_T$, $\hat{P}_1$, $\hat{P}_2$, $h_1$, $h_2)$ to $\B$. In addition, $\mathcal{C}$ also generates a Bloom filter as $(\{H_j\}_{j\in [k]}, T) \leftarrow \BF.\gen(\ell, k)$, and outputs the public key $\vk=(\hat{P}_{pub}, g, \{H_j\}_{j\in [k]})$ and challenging position $i^*$ to $\B$, where $g=e(\hat{P}_1,\hat{P}_{pub})$.
Next, $\mathcal{C}$ is ready to answer $\B$'s queries during the simulation.  For simplicity, we assume as in ~\cite{PKC:ChoChe03} that for any $i \in [l]$, $\B$ queries $h_1(i)$ at most once and any query involving $i$ is proceeded by the RO query $h_1(i)$, which means that $\B$ has to query $h_1(i)$ before he can obtain the signature ($h,S,i$) from \textbf{Signature} query and the secret key $sk_i$ from \textbf{Corruption} query by using a simple wrapper of $\B$.


  \item {\bf Hash function queries.} $\mathcal{C}$ initializes a counter $t$ to 1.
 \begin{itemize}[leftmargin=*]
  \item $h_1: [l] \rightarrow \mathbb{Z}^*_p$: On input $i \in [l]$, $\mathcal{C}$ returns a random $w = w^*\xleftarrow{R}\mathbb{Z}^*_p$ if $i=i^*$. Otherwise, $\mathcal{C}$ returns $w = w_t$ and sets $t = t+1$. In both cases, $\mathcal{C}$ stores ($i,w$) in a list $L_1$.

  \item $h_2: \{0, 1\}^* \times \G_T \rightarrow \Z^*_p:$ On input $(m, r)$, $\mathcal{C}$ outputs $h$ if $(m, r, h)$ is in the list $L_2$ (initialized to be empty). Otherwise, it outputs a random element $h$ and stores ($m,r,h$) in $L_2$.
 \end{itemize}
      Note that, according to the query of $h_1$ and the computation in the setup phase, $\mathcal{C}$ knows the secret key $sk_i = \frac{\alpha}{\alpha + w_i}\hat{P_1}$ for $i \neq i^*$.

 \item {\bf Query Phase.} $\mathcal{C}$ answers adaptive signing and puncturing queries from $\B$ as follows:
 \begin{itemize}[leftmargin=*]
  \item {\bf Signature query}: On input a message $m$ with prefix $m'$, $\mathcal{C}$ first checks $\BF.\check(\{H_j\}_{j\in [k]}, T, m') = 1$ and outputs $\bot$ in this case. Otherwise, there exists at least one index $i_j \in \{i_1, \ldots, i_k\}$ such that $\sk_{i_j} \neq \bot$, where $i_j = H_j(m')$, for $j \in [k]$. $\mathcal{C}$ picks a random $j \in \{i_1, \ldots, i_k\}$, $S \xleftarrow{R}  G$ and $h \xleftarrow{R}  \Z^*_p$, computes and sets $r = e(S,h_1(j)\hat{P}_2+\hat{P}_{pub}).e(\hat{P}_1,\hat{P}_{pub})^h$, and backpatches to define $h_2(m, r) = h$. Finally, $\mathcal{C}$ stores $(m, r, h)$ in $L_2$, and returns the signature $\sigma = (h, S, j)$ ($\mathcal{C}$ aborts in the unlikely collision event that $h_2(m, r)$ is already defined by other query results of $\sign$ or $h_2$, the probability of which is negligible since $r$ is random~\cite{JC:PoiSte00}).

  \item {\bf Puncture query}: On input a string $str$, $\mathcal{C}$ first updates $T' = \BF.\update(\{H_j\}_{j\in [k]}, T, str)$. Then for each $i \in [\ell]$, update
$$
\sk'_i = \left\{
             \begin{array}{ll}
             \sk_i, & \text{if} \ T'[i] = 0\\
             \bot, & \text{otherwise}\\
             \end{array}
\right.
$$
The updated signing key is $\sk' = (T', \{\sk'_i\}_{i\in [\ell]})$.
 \end{itemize}

 \ignore{\item \bf{Challenge Phase}: On input challenge message prefix $m^*$, $\mathcal{C}$ aborts the simulation if $i^* \in S_{m^*}$, where for any $i \in S_{m^*}$, we have $T[i] = 0$ for $T = \BF.\update(\{H_j\}_{j\in [k]}, T, m^*)$. Otherwise, puncture $m^*$ as described above.}

\item \textbf{Corruption query}: For {\sf Corruption} query, $\mathcal{C}$ recovers the matching pair ($i,w$) from $L_1$ and returns the previously computed $\frac{\alpha}{\alpha + w}\hat{P}_1$.

\item {\bf Forgery.} If adversary $\B$ forges a valid tuple $(m, r, h, S, i^*)$ in a time $t_1$ with probability $\epsilon_1 \geq 10(q_S+1)(q_S+q_{h_2})/p$, where the message $m$ has prefix $m'$, according to the forking lemma~\cite{JC:PoiSte00}, $\mathcal{C}$ can replay adversary $\B$ with different choices of random elements for hash function $h_2$ and obtain two valid tuples $(m, r, h', S_1, i^*)$ and $(m, r, h'', S_2, i^*)$, with $h'$$\neq$$h''$ in expected time $t_2$$\leq$$120686q_{h_2}t_1/{\epsilon_1}$.
\end{itemize}

Now we can apply the standard argument for outputs of the forking lemma as follows: $\mathcal{C}$ recovers the pair ($i^*,w^*$) from $L_1$, and note that $w^* \neq w_1,...,w_{\tau-1}$ with probability at least $1-{\tau}/{p}$. Since both forgeries are valid signatures, we can obtain the following relations:
$$e(S_1,Q_{i^*}).e(\hat{P}_1,\hat{P}_{pub})^{h'}=e(S_2,Q_{i^*}).e(\hat{P}_1,\hat{P}_{pub})^{h''},$$
where $Q_{i^*}=h_1(i^*)\hat{P}_2+\hat{P}_{pub}=(w^*+\alpha)\hat{P}_2$. Then we have
$$e((h''-h')^{-1}\cdot(S_1-S_2),Q_{i^*})=e(\hat{P}_1,\hat{P}_{pub}),$$
and hence $T^*=(\hat{P}_1-(h''-h')^{-1}.(S_1-S_2))/w^*=(\hat{P}_1-\frac{\alpha}{w^*+\alpha}\hat{P}_1)/w^*=\frac{1}{w^*+\alpha}\hat{P}_1=\frac{f(\alpha)}{w^*+\alpha}P_1$. From $T^*$, $\mathcal{C}$ can proceed as in~\cite{boneh2004short} to extract $\sigma^*=\frac{1}{w^*+\alpha}P_1$: $\mathcal{C}$ first writes the polynomial $f$ as $f(y)=\gamma(y)(y+w^*) + \gamma_{-1}$  by using the long division method, where $\gamma(y)=\sum_{i=0}^{\tau-2}\gamma_iy^i$ and $\gamma_{-1} \in \Z^*_p$, and eventually computes
$$\sigma^*=\frac{1}{\gamma_{-1}}\big[T^*-\sum_{i=0}^{\tau-2}\gamma_i\psi(\alpha^iP_2)\big]=\frac{1}{w^*+\alpha}P_1$$
and returns $(w^*, \sigma^*)$ as the solution to the $\tau$-SDH instance.
\ignore{From the description of reduction, it is straightforward to verify that the signing, puncturing and corruption queries produce valid outputs if $\mathcal{C}$ does not abort the simulation. Furthermore, since the outputs of signing and hash function queries, and the secret key returned by the corruption oracle are indistinguishable from that of the original scheme, adversary $\mathcal{C}$ learns nothing from query results. In addition, the collision event occurs only with negligible advantage due to $r$ is randomly generated as computed explicitly in ~\cite{JC:PoiSte00}.}
\end{proof}

\noindent The combination of the above lemmas yields Theorem 1.
\end{proof}

\section{Puncturable Signature in Proof-of-Stake Blockchain}

Before describing the application of the puncturable signature scheme in proof-of-stake blockchain, we recall some basic definitions~\cite{david2018ouroboros}\cite{kiayias2017ouroboros} of proof-of-stake blockchain and secure properties~\cite{kiayias2015speed}\cite{pass2017analysis} of blockchain. We assume that there are $n$ stakeholders $U_1, \ldots, U_n$, and each stakeholder $U_i$ possesses $s_i$ stake and a public/secret key pair $(\vk_i, \sk_i)$. Without loss of generality, we assume that all system users know the public keys $\vk_1, \ldots, \vk_n$. The protocol execution is divided in time units (also called slots).

\begin{definition}[Epoch, Block, State, Genesis Block, Blockchain~\cite{david2018ouroboros}] \label{defn:block}
An epoch is a collection of $R$ sequential slots $S=\{sl_1,...,sl_R\}$, where $R$ is a system parameter and the stake distribution for selecting slot leaders remains unchanged during one epoch.

A block generated at the slot $sl_j$ ($j \in [R]$) is a tuple of the form $B_j = (sl_j,st_j,d_j,\sigma_j)$\footnote{Recall that the puncturable signature proposed in this paper supports puncturing at any position. For ease of presentation, slot number $sl_j$ is defined as the prefix of the block, however, it maybe has different locations in specific PoS protocols.}, where $st_j \in \{0,1\}^\lambda$ denotes the state of the previous block $B_{j-1}$ such that $st_j = H(B_{j-1})$, $d_j \in \{0,1\}^*$ denotes the transaction data, and $\sigma_j$ denotes a signature on ($sl_j,st_j,d_j$) computed under the signing key of the stakeholder $U_i$ (i.e., the leader of the slot $sl_j$ who is eligible to generate the block).

The genesis block is denoted by $B_0 = (\mathbb{R}_0, \rho)$, where $\mathbb{R}_0 = \big((\vk_1,s_1),(\vk_2,s_2),...,(\vk_n,s_n)\big)$ contains the corresponding public keys and stakes of the stakeholders, and $\rho$ denotes the auxiliary information used as the seed for the process of the slot leader election.

A blockchain $\mathcal{C}$ is initialised as the genesis block $B_0$ followed by a sequence of blocks $B_1, \cdots, B_n$, where the block $B_j$ ($j \in [n]$) is generated in the slot $sl_j$ in an ascending order and the length of $\mathcal{C}$ (denoted by len$(\mathcal{C})$) is $n$. We call the rightmost block $B_n$ the head of $\mathcal{C}$, denoted by head$(\mathcal{C})$. \ignore{ We treat the empty string $\varepsilon$ as a legal chain and by convention set head$(\varepsilon) = \varepsilon$.}
\end{definition}

\begin{definition}[Properties of Blockchain~\cite{david2018ouroboros}] \label{defn:propertiesblockchain}
A blockchain protocol should satisfy the following three properties.
\begin{itemize}[leftmargin=*]
\item \textbf{Common Prefix.} $\mathcal{C}_i^{\lceil k}$ is the prefix of $\mathcal{C}_j$ if $\mathcal{C}_i$ and $\mathcal{C}_j$ are the chains maintained by two honest stakeholders at the slot $sl_i$ and  $sl_j$ ($sl_i < sl_j$) respectively, where $\mathcal{C}_i^{\lceil k}$ denotes the chain after pruning the last $k$ blocks from $\mathcal{C}_i$ and $k \in \mathbb{N}$ is the common prefix parameter.

\item \textbf{Chain Quality.} For any $l$ consecutive blocks in the chain of an honest stakeholder, the fraction of blocks contributed by the adversary is at most $1-\mu$. $\mu \in [0,1]$ is also called the chain quality coefficient.

\item \textbf{Chain Growth.} For two chains $\mathcal{C}_i$ and $\mathcal{C}_j$ possessed by two honest stakeholders at the slot $sl_i$ and $sl_j$ respectively, where $sl_j - sl_i \geq s$, we have len$(\mathcal{C}_j)$ $-$ len$(\mathcal{C}_i)$ $\geq$ $\tau \cdot s$. $\tau$ is also called the speed coefficient.
\end{itemize}
\end{definition}

\subsection{Application in Ouroboros Paros Protocol}

\begin{table*}[t]
\center
\begin{tabular}{|p{17.5cm}|}
\hline
\begin{center}
{\sf Functionality} $\mathcal{F}_{\textsf{PS}}$
\end{center}

\begin{description}[leftmargin=*]
\item $\mathcal{F}_{\textsf{PS}}$ runs between a specific signer $U_S$ and other stakeholders $U_1,...,U_n$ as follows:

\item[Key Generation.] When receiving $(\textsf{KeyGen}, sid, U_S)$ from some $U_S$, it first checks whether there exists some $sid'$ such that $sid = (U_S, sid')$ and aborts if not. Otherwise, it forwards $(\textsf {KeyGen}, sid, U_S)$ to the adversary $\mathcal{A}$. Upon receiving the response $(\textsf {PublicKey}, sid, U_S, vk)$ from $\mathcal{A}$, it outputs the public key $(\textsf {PublicKey}, sid, vk)$ to $U_S$, keeps the entry $(sid, U_S, vk)$, and initializes the set $P = \emptyset$.

\item[Sign and Puncture.] When receiving $(\textsf{PSign}, sid, U_S, m = m'...)$ from $U_S$, it first checks whether there exists a recorded entry $(sid, U_S, vk)$  for some $sid$ and $m' \notin P$, and aborts if not. Otherwise, it forwards $(\textsf {Sign}, sid, U_S, m)$ to the adversary $\mathcal{A}$. When receiving the output $(\textsf {Signature}, sid, U_S, m, \sigma)$ from $\mathcal{A}$, it checks whether there exists a recorded entry $(m, \sigma, vk, 0)$ and aborts with an error output if not. Otherwise, it outputs the signature $(\textsf{Signature}, sid, m, \sigma)$ to $U_S$, keeps the entry $(m, \sigma, vk, 1)$ recorded, and sets $P = P \cup \{m'\}$.

\item[Signature Verification.] When receiving $(\textsf{Verify}, sid, m=m'..., \sigma, vk')$ from some party $U_i$ ($i\in[n]$), it proceeds as follows:

\begin{enumerate}
\item If $vk' = vk$ and there exists a recorded entry $(m, \sigma, vk, 1)$ for $vk$, it sets $l = 1$. (The completeness property is guaranteed by this condition: If the public key has been registered and the signature on the msessage $m$ is generated legitimately, then the verification succeeds.)

\item Else, if $vk' = vk$, there is no recorded entry $(m, \sigma', vk', 1)$ for any $\sigma'$, and the signer is uncorrupted, then set $l = 0$ and keep the entry $(m, \sigma, vk', 0)$ recorded. (The unforgeability property is guaranteed by this condition: If the public key has been registered, the message $m$ has never been signed by the signer, and the signer is uncorrupted, then the verification fails.)

\item Else, if there exists a recorded entry $(m, \sigma, vk', l')$, then set $l = l'$. (The consistency property is guaranteed by this condition: the answers are identical for the repeated verification requests.)

\item Else, if $m' \in P$, then let $l = 0$ and keep the entry $(m, \sigma, v, 0)$ recorded. Otherwise, forward $(\textsf{Verify}, sid, m, \sigma, v')$ to $\mathcal{A}$. Upon receiving the response $(\textsf{Verified}, sid, m, \omega)$ from $\mathcal{A}$, set $l = \omega$ and keep the entry $(m, \sigma, v', \omega)$ recorded. (This condition ensures that the adversary $\mathcal{A}$ can only forge signatures of corrupted stakeholders on messages with the unpunctured prefix.)
\end{enumerate}
Finally, it forwards the output $(\textsf{Verified}, sid, m, l)$ to $U_i$.
\end{description}\\
\hline
\end{tabular}
\vspace{0.2cm}\\
{\bf{Figure 2: Functionality $\mathcal{F}_{\textsf{PS}}$ }}\\
\end{table*}

\begin{table*}[t]
\center
\begin{tabular}{|p{17.5cm}|}
\hline
\begin{center}
{\sf Protocol} $\pi'_{\textsf{SPoS}}$
\end{center}

The PoS protocol $\pi'_{\text{SPoS}}$ is run by stakeholders $U_i$ ($i \in [n]$) over a sequence of slots $S = \{sl_1, \cdots, sl_R\}$, and the sub-functionalities  $\mathcal{F}_{\textsf{INIT}}$, $\mathcal{F}_{\textsf{VRF}}$, $\mathcal{F}_{\textsf{PS}}$ and  $\mathcal{F}_{\textsf{DSIG}}$ are employed. Define $T_i =  2^{\ell_{\textsf{VRF}}}(1-(1-f)^{\alpha_i})$ as the threshold for a stakeholder $U_i$, where $\alpha_i$ is the relative stake owned by the stakeholder $U_i$, $\ell_{\textsf{VRF}}$ denotes the length of $\mathcal{F}_{\textsf{VRF}}$ output, and $f$ denotes the active coefficient. $\pi'_{\textsf{SPoS}}$ performs the following steps:

\begin{description}[leftmargin=*]
\item [Initialization.] Upon receiving the request (${\textsf {KeyGen}}, sid, U_i$) from the stakeholder $U_i$, $\mathcal{F}_{\textsf{PS}}$, $\mathcal{F}_{\textsf{DSIG}}$ and $\mathcal{F}_{\textsf{VRF}}$ return the response ({\sf{PublicKey}}, $sid, v_i^{\text{ps}}$), ({\sf{PublicKey}}, $sid, v_i^{\text{dsig}}$) and ({\sf{PublicKey}}, $sid, v_i^{\text{vrf}}$), respectively.

    In the case of the first round, $U_i$ forwards (${\textsf{ver\_keys}}$, $sid, U_i, v_i^{\text{ps}}, v_i^{\text{dsig}}, v_i^{\text{vrf}}$) to $\mathcal{F}_{\textsf{INIT}}$ for initial stake distribution. In any case, $U_i$ finally returns ($U_i, v_i^{\text{ps}}, v_i^{\text{dsig}}, v_i^{\text{vrf}}$) to $\mathcal{Z}$ and then terminates the round.
    In the next round, $U_i$ sends (${\textsf{genblock\_req}}$, $sid, U_i$) to $\mathcal{F}_{\textsf{INIT}}$ and receives ({\sf{genblock}}, $sid, \mathbb{B}_0, \eta$) as the response.

     If $U_i$ has been registered in the first round, the local chain of $U_i$ is set as $\mathcal{C} = B_0$ and its initial state is set as $st = H(B_0)$. Otherwise, the local chain $\mathcal{C}$ is provided by the environment and its initial state is set as $st = H({\text{head}}(\mathcal{C}))$.
\item [Chain Extension.] Each online stakeholder $U_i$ (for each slot $sl_j$) proceeds as follows:
\begin{enumerate}
\item $U_i$ collects the transaction data $d \in \{0, 1\}^*$.

\item $U_i$ collects all valid chains via broadcast and put them into a set $\mathbb{C}$. Then for each chain $\mathcal{C}' \in \mathbb{C}$ and each block $B' = $($st', d', sl', B'_{\pi}, \sigma_{j'}) \in \mathcal{C}'$, $U_i$ checks whether the stakeholder creating $B'$ belongs to the slot leader set for the slot $sl'$ (specifically, $U_i$ parses $B'_\pi$ as ($U_s, y', \pi'$) for some $s$, and verifies that ({\sf{Verified}}, $sid, \eta\| sl'$, $y', \pi', 1$) is returned by $\mathcal{F}_{\textsf{VRF}}$ for the request ({\sf{Verify}}, $sid, \eta\| sl'$, $y', \pi', v_s^{{\text{vrf}}}$) and moreover $y' < T_s$), and whether the request ({\sf{Verify}}, $sid, (sl', st', d', B'_{\pi}), \sigma_{j'}, v_s^{{\text{ps}}}$) to $\mathcal{F}_{\textsf{PS}}$ is responded by ({\sf{Verified}}, $sid, (sl', st', d', B'_{\pi}), 1$).
    $U_i$ sets the new local chain $\mathcal{C}'$ as the longest chain from $\mathbb{C}\cup\mathcal{C}$, and the state as $st = H({\text{head}}(\mathcal{C}'))$.

\item Upon receiving the request ({\sf{EvalProve}}, $sid, \eta\| sl_j$) from $U_i$, $\mathcal{F}_{\textsf{VRF}}$ returns ({\sf{Evaluated}}, $sid, y, \pi$). $U_i$ checks $y < T_i$ to judge whether it belongs to the slot leader set for the slot $sl_j$.
     If yes, a new block $B = (sl_j, st, d, B_\pi, \sigma)$ is produced by $U_i$, where $st$ is its current state, $d \in \{0, 1\}^*$ is the transaction data, $B_\pi = (U_i, y, \pi)$, and $\sigma$ is a signature obtained by receiving the response ({\sf{Signature}}, $sid, (sl_j, st, d, B_\pi), \sigma$) of $\mathcal{F}_{\textsf{PS}}$ for the request ({\sf{PSign}}, $sid, U_i, (sl_j, st, d, B_\pi)$). $U_i$ updates the new local chain to be $\mathcal{C}' = \mathcal{C}|B$, and sets the state $st = H({\text{head}}(\mathcal{C}'))$. Finally, $U_i$ broadcasts the new chain $\mathcal{C}'$ if he has created a new block in this step.
\end{enumerate}
\item [Transaction Generation.] When receiving the request (${\textsf{sign\_tx}}$, $sid', tx$) from the environment, $U_i$ forwards the request ({\sf{Sign}}, $sid, U_i, tx$) to $\mathcal{F}_{\textsf{DSIG}}$, receives the returned ({\sf{Signature}}, $sid, tx, \sigma$), and finally sends the message (${\textsf{signed\_tx}}$, $sid', tx, \sigma$) to the environment.
\end{description}\\
\hline
\end{tabular}
\vspace{0.2cm}\\
{\bf Figure 3: Protocol $\pi'_{\textsf{SPoS}}$ }\\
\end{table*}

Recall that in Ouroboros Praos protocl~\cite{david2018ouroboros}, the adversary is allowed to adaptively corrupt stakeholders under the honest majority of stake assumption, which is achieved by formalizing and realizing the forward secure signatures and the verifiable random function (VRF) under the UC framework.

In their security analysis~\cite{david2018ouroboros}, the basic PoS protocol for the static stake case is formalized into a functionality $\mathcal{F}_{\textsf{SPoS}}$ run by the stakeholders in the $\mathcal{F}_{\textsf{INIT}}$-hybrid model, where the stakeholders can interact with the ideal functionalities $\mathcal{F}_{\textsf{INIT}}$, $\mathcal{F}_{\textsf{VRF}}$, $\mathcal{F}_{\textsf{KES}}$ and $\mathcal{F}_{\textsf{DSIG}}$.  Specifically, the functionality $\mathcal{F}_{\textsf{INIT}}$ determines the genesis block $B_0$, where $\rho$ is set to be the random nonce $\eta$, and the public key for the stakeholder $U_i$ is $vk_i = \{v_i^{\text{vrf}}, v_i^{\text{kes}}, v_i^{\text{dsig}}\}$ for the functionality $\mathcal{F}_{\textsf{VRF}}$, $\mathcal{F}_{\textsf{KES}}$ and $\mathcal{F}_{\textsf{DSIG}}$ respectively. In addition, $\mathcal{F}_{\textsf{VRF}}$ models a special VRF used to select the leaders eligible to issue new blocks, $\mathcal{F}_{\textsf{KES}}$ models forward secure (key-evolving) signature schemes used by the leaders to sign the new block, and $\mathcal{F}_{\textsf{DSIG}}$ models EUF-CMA secure signature schemes used to sign the transactions. In addition, the block in Ouroboros Praos includes an additional entry $B_{\pi}$, which is the output of VRF and used as the block proof and the validity check. It is proved $\mathcal{F}_{\textsf{SPoS}}$ can achieve common prefix, chain growth and chain quality properties, and these properties remain unchanged when all of the underlying functionalities are replaced by their real world implementations (also called real experiment). Finally, the protocol is extended to the dynamic case allowing for changes of the stake distribution over time. To avoid unnecessary repetitions, we omit details of the above functionalities and proofs~\cite{david2018ouroboros}.


The puncturable signature can resist LRSL attack due to the fact that the leader $U$ in slot $sl_j$ would update the secret signing key $sk$ after the block is proposed, and with the updated signing key the adversary cannot forge a signature at $sl_j$ in the name of $U$ and thus cannot re-write a new block at the position $sl_j$. We now present an ideal functionality $\mathcal{F}_{\textsf{PS}}$ of puncturable signature scheme (described in Figure 2), replace $\mathcal{F}_{\textsf{KES}}$ using $\mathcal{F}_{\textsf{PS}}$ in the revised static proof-of-stake protocol $\pi'_{\text{SPoS}}$ (described in Figure 3), prove that $\pi'_{\text{SPoS}}$ still satisfies all properties of blockchain (i.e., common prefix, chain quality and chain growth)\footnote{Also note that as in \cite{david2018ouroboros}, we also assume in this paper that honest stakeholders would perform erasures securely to achieve key update, which is reasonable to capture protocol security against adaptive adversaries as argued in \cite{lindell2009adaptively}.}, and also show that $\mathcal{F}_{\textsf{PS}}$ can be realized by basic puncturable signature construction in Section 3.2.

\ignore{In order to provide security against long-range attack and allow the adversary to corrupt any stakeholders adaptively as long as the stakeholder distribution maintains an honest majority of stake, Ouroboros Paros~\cite{david2018ouroboros} introduces the ideal functionality $\mathcal{F}_{\textsf{KES}}$ of key-evolving signature scheme with forward security and constructs the proof-of-stake blockchain protocol where blocks are signed with a forward secure signature scheme modeled by $\mathcal{F}_{\textsf{KES}}$. They also provide combinatorial analysis of the proposed proof-of-stake protocol $\pi_{\text{DPoS}}$ using the idealized functionalities including $\mathcal{F}_{\textsf{KES}}$.}


In a high level, the ideal functionality $\mathcal{F}_{\textsf{PS}}$ (Figure 2) allows an adversary compromising the signing keys of the signer to forge signatures only for messages with the unpunctured prefix. Starting from the regular signature functionality in~\cite{canetti2004univerally}, $\mathcal{F}_{\textsf{PS}}$ is extended by packing the signing operation with a puncture operation. For the verification operation,  $\mathcal{F}_{\textsf{PS}}$ allows the adversary to control the response only for the signature on the messages with the unpunctured prefix.

\begin{theorem}\label{thm:KES-PS}
The revised proof-of-stake blockchain protocol $\pi'_{\text{SPoS}}$ described in Figure 3 still satisfies common prefix, chain quality and chain growth.
\end{theorem}

\ignore{
We present the full proof in Appendix \ref{appendix:functionreplace}. The strategy of our proof is: given the event of violating one of common prefix, chain quality and chain growth in an execution of $\pi'_{\text{SPoS}}$ with access to $\mathcal{F}_{\textsf{PS}}$ by adversary $\A$ and environment $\mathcal{Z}$, we can construct another adversary $\A'$  so that the corresponding event happens with the same probability in an execution of $\pi_{\text{SPoS}}$ with access to $\mathcal{F}_{\textsf{KES}}$ (c.f. Appendix \ref{appendix:KES function}) by adversary $\A'$ and environment $\mathcal{Z}$, where $\pi_{\text{SPoS}}$ is original protocol ~\cite{david2018ouroboros}. If the environment $\mathcal{Z}$ can distinguish a real execution with $\mathcal{A}$ and $\pi'_{\text{SPoS}}$ (accessing $\mathcal{F}_{\textsf{PS}}$) from an ideal execution, then $\mathcal{Z}$ can also distinguish a real execution with $\A'$ and $\pi_{\text{SPoS}}$  (accessing $\mathcal{F}_{\textsf{KES}}$) from an ideal execution.}

\begin{proof}\renewcommand{\qedsymbol}{$\blacksquare$}
 Given the event of violating one of common prefix, chain quality and chain growth in an execution of $\pi'_{\text{SPoS}}$ with access to $\mathcal{F}_{\textsf{PS}}$ by adversary $\A$ and environment $\mathcal{Z}$, we can construct an adversary $\A'$ so that the corresponding event happens with the same probability in an execution of $\pi_{\text{SPoS}}$ with access to $\mathcal{F}_{\textsf{KES}}$ by adversary $\A'$ and environment $\mathcal{Z}$, where $\pi_{\text{SPoS}}$ is original protocol~\cite{david2018ouroboros}. Specifically, the adversary $\A'$ simulates $\A$ as follows:
\begin{itemize}[leftmargin=*]
 \item Upon receiving $(\textsf {KeyGen}, sid, U_S)$ from $\mathcal{F}_{\textsf{PS}}$, $\A'$ runs as in the case of $\mathcal{F}_{\textsf{KES}}$ for key generation, sets counter $\mathsf{k}_\mathsf{ctr} = 1$ and $P = \emptyset$, and sends $(\textsf {PublicKey}, sid, U_S, v)$ to $\mathcal{F}_{\textsf{PS}}$\footnote{ In $\mathcal{F}_{\textsf{KES}}$\cite{david2018ouroboros}, $\mathsf{k}_\mathsf{ctr}$ denotes the current time period. Specifically, after signing any message in the period $\mathsf{k}_\mathsf{ctr}$, the signer would generate a new secret key for the next time period  $\mathsf{k}_\mathsf{ctr}+1$. } .
 \item Upon receiving $(\textsf {Sign}, sid, U_S, m = m'\cdots)$ from $\mathcal{F}_{\textsf{PS}}$, $\A'$ ignores the request if $m' \in P$. Otherwise, it sets $j = \mathsf{k}_\mathsf{ctr}$ and computes the signature $\sigma$ as in the case of $\mathcal{F}_{\textsf{KES}}$. Then $\A'$ updates the corresponding secret key, sets counter $\mathsf{k}_\mathsf{ctr} = j + 1$ and $P = P \cup {m'}$, and sends $(\textsf {Signature}, sid, U_S, m, \sigma)$ to $\mathcal{F}_{\textsf{PS}}$.
 \item Upon receiving $(\textsf{Verify}, sid, m, \sigma, v')$ from $\mathcal{F}_{\textsf{PS}}$, $\A'$ verifies the signature as in the case of $\mathcal{F}_{\textsf{KES}}$, and sends $(\textsf{Verified}, sid, m, \phi)$ to $\mathcal{F}_{\textsf{PS}}$.
\end{itemize}

Note that in an execution of $\pi'_{\text{SPoS}}$ with access to $\mathcal{F}_{\textsf{PS}}$, $m'$ in $\mathcal{F}_{\textsf{PS}}$ equals $sl$ (i.e, the slot parameter of the last block) (c.f. Definition~\ref{defn:block}), while in the execution of $\pi_{\text{SPoS}}$ with access to $\mathcal{F}_{\textsf{KES}}$, the input to signature algorithm is $(\textsf{Usign}, sid, m=sl||..., sl)$, which means that the update of punctured set $P$ is consistent with that of counter $\mathsf{k}_\mathsf{ctr}$. In other words, when one signing happens on $m$ containing some prefix $sl$, $P$ adds $sl$ in $\mathcal{F}_{\textsf{PS}}$ while $\mathsf{k}_\mathsf{ctr}$ increases by 1 in $\mathcal{F}_{\textsf{KES}}$.

Therefore, $\mathcal{A}'$ can simulate the execution for $\A$. If the environment $\mathcal{Z}$ can distinguish a real execution with $\mathcal{A}$ and $\pi'_{\text{SPoS}}$ (accessing $\mathcal{F}_{\textsf{PS}}$) from an ideal execution  that provides the properties of common prefix, chain quality and chain growth, then $\mathcal{Z}$ can also distinguish a real execution with $\A'$ and $\pi_{\text{SPoS}}$ (accessing $\mathcal{F}_{\textsf{KES}}$) from an ideal execution, which means that any winning advantage of the adversary against common prefix, chain quality and chain growth in $\pi'_{\text{SPoS}}$ with access to $\mathcal{F}_{\textsf{PS}}$ immediately implies at least the same advantage in $\pi_{\text{SPoS}}$ with access to $\mathcal{F}_{\textsf{KES}}$.
\end{proof}

\begin{remark}
\emph{The dynamic stake case can be extended just as in \cite{david2018ouroboros}. Specifically, $\mathcal{F}_{\textsf{INIT}}$ is replaced with a resettable leaky beacon functionality, where the adversary is allowed to obtain the nonce value for the next epoch and reset the nonce several times, and other sub-functionalities remain unchanged.}
\end{remark}

{\it Realizing $\mathcal{F}_{\textsf{PS}}$.} Following the proof strategy of ~\cite{canetti2004univerally}, we show how to translate a puncturable signature scheme $\Sigma$ into a signature protocol $\pi_{\Sigma}$ in the present setting and then prove that $\pi_{\Sigma}$ can securely realize $\mathcal{F}_{\textsf{PS}}$. Specifically, $\pi_{\Sigma}$ protocol runs between a signer $U_{S}$ and other stakeholders $U_1,...,U_n$, and proceeds based on a puncturable signature scheme $\Sigma$=($\setup$, $\punc$,  $\sign$, $\verify$) as follows:

\begin{enumerate}[leftmargin=*]
\item {\bf Key Generation:} When running $\pi_{\Sigma}$ and receiving ($\textsf{KeyGen}, sid, U_S$), $U_S$ first checks whether there exists $sid'$ such that $sid = (U_S, sid')$. If yes, he runs $\setup(1^{\lambda})$ to generate the public/secret key pair $(vk, sk)$, keeps the signing key ($sid, U_S, sk$) for $U_S$, sets $P = \emptyset$, and finally outputs ($\textsf{PublicKey}, sid, vk$). Otherwise, it aborts.

\item {\bf Sign and Puncture:} When receiving ($\textsf{PSign}$, $sid$, $U_S$, $m=m'...$) such that $U_S$ owns the signing key ($sid$, $U_S$, $sk$) for $sid$, $U_S$ checks whether $m' \in P$. If not, $U_S$ runs $\sign(sk,m)$ to obtain the signature $\sigma$, runs $\punc(sk,m')$ to update the secret key, sets $P=P \cup{m'}$, and outputs ($\textsf{Signature}, sid, m, \sigma$).

\item {\bf Verify:} If stakeholder $U_i$ ($i\in[n]$) receives ($\textsf{Verify}, sid, m, \sigma, vk'$) as an input, he outputs ($\textsf{Verified}, sid, m, \verify(vk',m,\sigma)$).
\end{enumerate}

\begin{theorem}\label{thm:realizeps}
If a puncturable signature scheme $\Sigma = (\setup, \punc, \sign, \verify)$ satisfies the unforgeability with adaptive puncturing as in Definition \ref{defn:unforg}, then $\pi_{\Sigma}$ securely realizes $\mathcal{F}_{\textsf{PS}}$.
\end{theorem}


\begin{proof}\renewcommand{\qedsymbol}{$\blacksquare$}
Assume that $\pi_{\Sigma}$ does not realize $\mathcal{F}_{\textsf{PS}}$, i.e. for any ideal-process simulator $\mathcal{S}$， there exists an environment $\mathcal{Z}$ that can distinguish a real setting with $\mathcal{A}$ and  $\pi_{\Sigma}$ from an ideal setting with $\mathcal{S}$ and $\mathcal{F}_{\textsf{PS}}$. Then following the proof approach of ~\cite{canetti2004univerally}, we use $\mathcal{Z}$ to construct a forger $G$ that wins with non-negligible probability in the experiment $\expt^{\sf ps}_G(1^\secparam)$ for the underlying puncturable signature scheme $\Sigma$ as defined in Definition \ref{defn:unforg}, which in turn violates the unforgeability with adaptive puncturing of $\Sigma$. Since $\mathcal{Z}$ can succeed for any simulator $\mathcal{S}$, it also succeeds for the following specific $\mathcal{S}$, where $\mathcal{S}$ runs a simulated copy of $\mathcal{A}$:

\begin{enumerate}[leftmargin=*]
\item $\mathcal{S}$ forwards any input from $\mathcal{Z}$ to $\mathcal{A}$, and also forwards any outputs from $\mathcal{A}$ to $\mathcal{Z}$.

\item Whenever receiving ($\textsf{KeyGen},sid,U_S$) from $\mathcal{F}_{\textsf{PS}}$, $\mathcal{S}$ checks whether $sid = (U_S, sid')$ for some $sid'$, and ignores this request if not. Otherwise, $\mathcal{S}$ runs $\setup(1^{\lambda})$ to generate the public/secret key pair $(vk,sk)$, keeps the signing key ($sid,U_S,sk$) recorded, sets $P = \emptyset$, and returns ($\textsf{PublicKey},sid,vk$) to $\mathcal{F}_{\textsf{PS}}$.

\item Whenever receiving ($\textsf{PSign},sid,U_S,m=m'...$) from $\mathcal{F}_{\textsf{PS}}$, $\mathcal{S}$ checks whether there exists a recorded signing key ($sid,U_S,sk$) and $m' \notin P$, and ignores this request if not. Otherwise, $\mathcal{S}$ runs $\sign(sk,m)$ to obtain $\sigma$, runs $\punc(sk,m')$ to obtain the update secret keys, sets $P = P \cup \{m'\}$, and returns ($\textsf{Signature}$, $sid$, $m$, $\sigma$) to $\mathcal{F}_{\textsf{PS}}$.

\item Whenever receiving ($\textsf{Verify},sid,m,\sigma,vk'$) from $\mathcal{F}_{\textsf{PS}}$,  $\mathcal{S}$ returns ($\textsf{Verified},sid,m,\verify(vk',m,\sigma)$) to $\mathcal{F}_{\textsf{PS}}$.

\item Whenever $\mathcal{A}$ corrupts a party $U_i$,  $\mathcal{S}$ corrupts $U_i$ in the ideal execution. If $U_i$ is the signer $U_S$, $\mathcal{S}$ returns the current signing keys $sk$ as well as the internal state (possibly empty) of the signing algorithm $\sign$ as the internal state of $U_i$.
\end{enumerate}

Recall that $\mathcal{Z}$ can distinguish the real execution with $\mathcal{A}$ and $\pi_{\Sigma}$ from an ideal execution with $\mathcal{S}$ and $\mathcal{F}_{\textsf{PS}}$, then we demonstrate that the underlying $\Sigma$ is forgeable by constructing a forger $G$ as follows. $G$ runs a internal instance of $\mathcal{Z}$, and simulates the interactions with $\mathcal{S}$ and $\mathcal{F}_{\textsf{PS}}$ for $\mathcal{Z}$, where $G$ acts as both $\mathcal{S}$ and $\mathcal{F}_{\textsf{PS}}$ to $\mathcal{Z}$. Moreover, in the simulating process, like the case for $\mathcal{S}$, $G$ will also run a simulated copy of $\mathcal{A}$.

When $G$ is activated, it returns the public key $vk$ from its own experiment to $\mathcal{A}$. When $\mathcal{Z}$ activates $U_S$ with input ($\sign,sid,U_S,m=m'...$), $G$ calls its signing oracle with $m$ to obtain a signature $\sigma$, calls its puncture oracle with $m'$ to update the secret keys, and then updates the puncturing set $P = P \cup \{m'\}$ and the set of queried messages $Q_{\sf sig} = Q_{\sf sig} \cup \{m\}$. When $\mathcal{Z}$ activates an uncorrupted party with input ($\textsf{Verify},sid,m=m'...,\sigma,vk'$), $G$ checks whether $m \in Q_{\sf sig}$, the signer is uncorrupted before $m'$ is punctured, and $\verify$($vk'$,$m$,$\sigma$) = $1$. If all these conditions are met, then in its experiment $\expt^{\sf ps}_G(1^\secparam)$, $G$ outputs $m'$ as the challenging string, and makes series of queries as in Definition \ref{defn:unforg}. Eventually $G$ outputs the tuple ($m,\sigma$), succeeding in the experiment.

Denote by $E$ the event that in a run of $\pi_{\Sigma}$ with $\mathcal{Z}$, some party $U_i$ is activated with a verification request ($\verify,sid,m$=$m'...,\sigma,vk'$), where $\verify(vk',m,\sigma) = 1$, $m \notin Q_{\sf sig}$, and $U_S$ is not corrupted before $m'$ is punctured. If $E$ does not occur, $\mathcal{Z}$ would not distinguish between a real and an ideal execution. However, we assume that $\mathcal{Z}$ can distinguish real from ideal executions with non-negligible advantage, and thus $E$ also happens with non-negligible advantage. Note that, from the view of $\mathcal{Z}$ and $\mathcal{A}$, the interaction with $G$ looks the same as the interaction with $\pi_{\Sigma}$, which means that whenever $E$ happens, a valid forgery is outputted by $G$.
\end{proof}

\subsection{Applications in Other Proof-of-Stake Protocols}

As we have described above, most existing proof-of-stake blockchain protocols are vulnerable to LRSL attack, and we show that our puncturable signature construction can also be applied in other Proof-of-stake blockchain protocols to resist LRSL attack.

 In both Ouroboros~\cite{kiayias2017ouroboros} and Snow White~\cite{daian2016snow} protocols, each block is signed by the leader using an ordinary signature scheme and thus they cannot resist LRSL attack. Fortunately, their signature schemes can also be replaced by the puncturable signatures directly. Specifically, in Ouroborous, the leader $U_i$ signs the block $B_i$ by $\sigma = \sign(\sk_j, (sl_j, d, st_j))$ and updates the secret key of $U_i$ by ${\textsf{Puncture}}(\sk_j, sl_j)$, and the case in Snow White is similar with the exception that the slot parameter is replaced with the time step $t$. By this means, even if an adversary $\A$ obtains the updated secret key, he cannot sign for other block data $d'$ at the same slot $sl_i$ or time step $t$, which furthermore avoids the forks in blockchains and LRSL attack. In addition, our puncturable signature also can be applied in Ouroboros Genesis \cite{badertscher2018ouroboros} protocol similar to Section 4.1.

\subsection{On tolerating a non-negligible correctness error for Proof-of-Stake Blockchain}
The significant efficiency improvement of our PS construction stems partially from the relaxation of tolerating a non-negligible correctness error, which, in turn, comes from the non-negligible false-positive probability of a Bloom filter. Specifically, the correctness error in our puncturable signature construction means that the signing of a message $m$ may yield $\bot$ even though the secret key has never been punctured at the prefix $m'$ of the message $m$. However, the correctness error can be as small as possible by adjusting the corresponding parameters in Bloom filter (see Section 2.1), which implies a trade-off between the non-negligible correctness error and the size of secret keys.

For proof-of-stake blockchain, it is a reasonable approach to accept a small, but non-negligible correctness error, in exchange for the huge efficiency gain. In fact, existent blockchain protocols achieve security properties (i.e. common prefix, chain quality and chain growth) with high probability instead of certainty, which means a small error probability is inherent in these protocols. Moreover, the signing error would not affect the running of the blockchain system. For instance, in Ouroboros~\cite{kiayias2017ouroboros}, the stakeholder selected as one of the leaders in current slot can still get the reward even if his signing fails. While in Ouroboros Praos~\cite{david2018ouroboros} and Snow White~\cite{daian2016snow}, some slots might have multiple slot leaders or no leader (i.e., empty slot), which means the signing error for one leader would not affect the protocol running.

\begin{table*}[!t]
\center
\begin{threeparttable}
\caption{Efficiency comparison}
\begin{tabular}{c|c|c|c}
\hline     &  Ours&  \cite{itkis2001forward}&  \cite{malkin2002efficient}\\
\hline      Keygen time&   $l \cdot t_{m1}$  &  $T\lambda t_{pt} + T \cdot t'_{mN}  + 3 t_{eN}$&   $\lambda t_{eN}$ \\
\hline      Sign time&   $t_{eT}+t_{m1}+k.t'_h+t_h$ &  2 $t_{eN}$ + $t_{mN}$ + $t_h$&  2 $t_{eN}$ + $t_{mN}$ + $t_h$   \\
\hline      Verify time&   $k.t'_h+t_h+t_p+t_{eT}+t_{m2}$  &  2 $t_{eN}$ + $t_{mN}$ + $t_h$&  (4 $t_{eN}$ + 2 $t_{mN}$) + ({\text {log}}$\lambda$+{\text {log}}$t$) $ t_h$ \\
\hline      Key update time&  $k.t'_h$  &  $T \cdot t_{eN}$ +  $T\lambda \cdot t_{pt}$&   $t \cdot t_{eN}$\\
\hline      Secret key size&   $l \cdot e^{-k|P|/l}|\mathbb{G}_1|$   & $3|\mathbb{Z}^*_{N}| + \lambda + 2{\text {log}T}$&  $|\mathbb{Z}^*_{N}| + \lambda \cdot ({\text {log}\lambda}+\text{log}t)$\\
\hline      Public key size&  $|\mathbb{G}_1|$    &  $2|\mathbb{Z}^*_{N}| + {\text {log}T}$&  $\lambda$\\
\hline      Signature size&  $|\mathbb{Z}^*_{p}|+|\mathbb{G}_1|$  &   $|\mathbb{Z}^*_{N}| + 2\lambda + {\text {log}T}$&  $4|\mathbb{Z}^*_{N}| + \lambda  ({\text {log}\lambda}+{\text {log}t})$ \\
\hline
\end{tabular}
\ignore{\begin{tablenotes}
        \footnotesize
        \item[1] The time cost for the generation of the initial secret signing key from public parameters.
      \end{tablenotes}}
\end{threeparttable}
\end{table*}

\begin{table*}[!t]
\center
\caption{Experimental results comparison}
\begin{threeparttable}
\begin{tabular}{c|c|c|c|c|c|c}
\hline    &\multicolumn{2}{c|}{  Ours} &\multicolumn{2}{c|}{  \cite{itkis2001forward}}&\multicolumn{2}{c}{  \cite{malkin2002efficient}}\\
\hline& 128-bit& 192-bit& 128-bit& 192-bit& 128-bit& 192-bit\\
\hline      Keygen time (ms)&  $5.19 \times 10^3$ &  $1.11 \times 10^4$ &  $6.92 \times 10^4$ &  $1.62 \times 10^5$ &  $378.88$ & $3.40 \times 10^3$\\
\hline      Sign time (ms)&  $1.17$ &  $5.61$ &  $5.92$ &  $35.41$ &  $5.92$ &  $35.41$\\
\hline      Verify time (ms)&  $4.14$ &  $23.09$ &  $5.92$ &  $35.41$ &  $11.84$ &  $70.82$\\
\hline      Key update time (ms)&  $10^{-5}$ &  $10^{-5}$ &  $3.65\times10^5$ &  $1.93\times10^6$ &  $2.96\times10^5$ &  $1.77\times10^6$\\
\hline      Secret key size\tnote{1}&  \tabincell{c}{$1.31 \times$\\$ e^{-|P|/{1.44\times10^3}} M$} &  \tabincell{c}{$1.64 \times$\\ $e^{-|P|/{1.44\times10^3}} M$} &  $1.14 KB$ &  $2.84 KB$ &  $761.75 B$ &  $1.51 KB$\\
\hline     Public key size&  $95.25B$ &  $119.50 B$ &  $770.08B$ &  $1.88 KB$ &  $16 B$ &  $24 B$\\
\hline     Signature size&  $129.11 B$ &  $169.35 B$ &  $418.08 B$ &  $1010.08 B$ &  $1.87 KB$ &  $4.32 KB$\\
\hline
\end{tabular}
\begin{tablenotes}
        \footnotesize
        \item[1] The secret size decreases with puncturing operation. For 128-bit and 192-bit security level, the maximum is $1.31 MB$ and $1.64 M$ when $|P| = 0$, and the minimum is $0.65 M$ and $0.82 M$ when $|P|$ reaches its maximum (i.e., $10^3$), respectively.
      \end{tablenotes}

\end{threeparttable}

\end{table*}

\begin{table*}[!t]
\newcommand{\minitab}[2][l]{\begin{tabular}{#1}#2\end{tabular}}
\center
\caption{Experimental cost of each unit operation (ms)}
\begin{tabular}{c|c|c|c|c|c|c|c|c|c|c}
\hline      &$t_{m1}$ &$t_{m2}$ &$t_{eT}$ &$t_p$ &$t_{mN}$ &$t'_{mN}$ &$t_{eN}$ &$t_{pt}$&$t_h$ &$t'_h$\\
\hline 128-bit &$0.36$ &$0.97$ &$0.81$ &$2.36$ &$0.0035$ &$0.00095$ &$2.96$ &$0.0054$&\multirow{2}{*}{$2 \times 10^{-5}$}&\multirow{2}{*}{$10^{-6}$}\\
\cline{1-9} 192-bit &$0.77$ &$6.93$ &$4.84$ &$11.32$ &$0.011$ &$0.0023$ &$17.7$ &$0.0084$&&\\
\hline
\end{tabular}

\end{table*}
\subsection{Analysis and Comparison}

For the proof-of-stake blockchain application, we make a comparison between our puncturable signature and two existing forward secure signatures, in terms of functionality and performance. First, puncturable signature allows each leader to generate at most one block at any slot (by puncturing at $sl_i$, the slot number of the current block), and thus prevents attackers from compromising leaders to mount LRSL attack. Although the forward secure signature can achieve the same functionality by using different secret key for signing in each period, their performance depends on the number of time periods $T$ (being set in advance) or the time periods $t$ elapsed so far, which is undesirable for the blockchain application. More specifically, in each slot of the proof-of-stake blockchain, only one stakeholder is elected as the leader to propose and sign the block, which means some stakeholders may only have a chance to sign block after long slots (i.e. time periods), however, the computational cost of one signature in this case may almost amount to that of multiple regular signatures due to the dependence on $T$ or $t$. On the contrary, the puncturable signature can alleviate this problem because the computation is independent of time periods.

Second, keeping on the signing and verifying operation as efficient as the underlying scheme is an important goal for the forward-secure signatures as well as our puncturable signatures. However, except for \cite{itkis2001forward}, almost all existing forward secure signature schemes require longer time for signing or verifying. Particularly, \cite{krawczyk2000simple}\cite{malkin2002efficient} requires two ordinary verification together with several hash computations, and verification time in \cite{abdalla2000new}\cite{bellare1999forward} even grows linearly with the number of periods $T$. Apparently, our construction can retain the efficiency of the underlying scheme on signing and verifying, with $k$ additional hash computations.

Third, the key update time of \cite{itkis2001forward}\cite{malkin2002efficient} depends on $T$ or $t$, which may bring some undesirable consumption and even become a fatal issue for some particular applications. Specifically, in the proof-of-stake blockchain, the signer may not even do any signing within one period but he has to update the signing key as long as the current period ends, which makes the update operation a vain effort. In some other applications, the signer has to update the secret key immediately after one signing operation, leading to that the number of update operations (i.e., $T$) within a given validity time of the public key becomes so large that the update time is unacceptable. \ignore{Despite that the general construction \cite{krawczyk2000simple} outperforms other schemes on key update, $O(T)$ non-secret certificates storage (in publicly readable tamper-proof memory) is needed, and moreover, one update includes complete key generation and verification process, which is also undesirable.} The key update in our puncturable signature construction is independent on $T$ or $t$, and only needs $k$ hash computations.

Finally, in Table 1 we compare the performance of our construction with that of \cite{itkis2001forward} and \cite{malkin2002efficient}, which are most efficient in existing forward secure signature schemes. We use $t_h$\ignore{\footnote{During the test to sign a string ``MESSAGE", we use the universal $t_h$ for timing evaluation of hash functions with different inputs in three schemes for simplicity, since the cost in three cases are almost equal (about $2 \times 10^{-5} ms$). Even though it can be ignored compared with the exponentiations and pairings according to our test, we keep it during the comparison due to the coefficient $\text{log}t$ as in Table 2.}}, $t'_h$, $t_{m1}$, $t_{m2}$, $t_{eT}$, $t_{p}$,  $t_{mN}$, $t'_{mN}$, $t_{eN}$ and  $t_{pt}$ to denote the time for computing a universal hash, a hash for $H_j(j \in [k])$ in Bloom filter, a multiplication in $\G_1$, a multiplication in $\G_2$, an exponentiation in $\G_T$, a bilinear pairing, a multiplication in $\Z^*_N$, a multiplication in $\varphi(N)$, an exponentiation in $\Z^*_N$ and one primality test for one $\lambda$-bit number, respectively. We also denote  $|\Z^*_p|$, $|\Z^*_N|$ and $|\mathbb{G}_1|$ as the bit-length of an element in $\Z^*_p$, an element in $\Z^*_N$ and an element in $\mathbb{G}_1$, respectively, where $p$ is the order of $\mathbb{G}_1$.


The implementations are written in C using version 3 of AMCL \cite{AMCL} and compiled using gcc 5.4.0, and the programme runs on a Lenovo ThinkCentre M8500t computer with Ubuntu 16.04.9 (64 bits) system, equipped with a 3.40 GHz Intel Core i7-4770 CPU with 8 cores and 8GB memory. Particularly, the AMCL library recommends two types of BLS curves (i.e., BLS12 and BLS24) to support bilinear pairings, and the curves have the form $y^2 = x^3 + b$ defined over a finite field $\mathbb{F}_q$, with $b = 15$ and $|q|=383$ for BLS12, while $b = 19$ and $|q|=479$ for BLS24, where $q$ is a prime. According to the analysis \cite{barbulescu2017updating}, BLS12 and BLS24 curves can provide 128-bit and 192-bit security levels respectively. For the group $\mathbb{Z}^*_N$, we choose $|N| = 3072$ and $|N| = 7680$ for 128-bit and 192-bit security levels respectively. For hash function, we choose SHA-384\footnote{The hash function $H_j(j \in [k])$ in bloom filter can be simulated by two hash functions according to the analysis in \cite{kirsch2008less}. In practice, the guava library \cite{guava} by Google employs Murmur3 hash \cite{murmur} for Bloom filter. For simplicity, we replace Murmur3 with SHA-384 during the test, however, our scheme would perform better using the faster Murmur3.}. In addition, we assume one stakeholder can be leader for $10^3$ times on average and set $n = 10^3$ in Bloom filter. Without loss of generality we assume the average probability that one stakeholder is selected as the leader in one slot is 1/100 (which is large enough in practice)\footnote{Note that we just choose the specific parameters to carry out the efficiency comparison. For larger $n$ and $T$, the efficiency of our scheme remains unchanged except that the time for key generation and secret size would increase according to Table 1, and thus the advantage of our scheme over forward secure signature schemes in the aspect of sign/verify/key update time as well as signature size still holds.}, which means there are at least $10^5$ slots in blockchain and set $T = 10^5$. We also set the error probability $pr = 1/1000$ of Bloom filter, then we can compute $\ell = -\frac{n\ln pr}{(\ln 2)^2} = 1.44 \times 10^4$ and $k = \lceil \frac{l}{n}\ln 2 \rceil = 10$. Note that $t$ in \cite{malkin2002efficient} denotes the time periods elapsed, also the number of signed operations so far, so we set $t = 10^5$ to evaluate the worst case.


Table 2 summarizes the experiment results, where the time represents the average time for 100 runs of each operation and the experimental cost of each basic operation over recommended groups at different security levels is shown in Table 3. The results show that our scheme performs better on signing and verification efficiency, significantly on key update efficiency. Moreover, our scheme has the smallest signature size, which drastically reduces the communication complexity for proof-of-stake blockchain. In addition, key generation in our scheme can be further optimized by pre-computing some exponentiations {\it off-line}. However, the initial secret key size in our scheme is large due to the Bloom filter. Fortunately, the secret key size shrinks with increasing amount of signing operations. In practice, the secret keys are stored locally on personal equipments, and reducing computation complexity and communication complexity may be more important with the rapid advance of storage technology.



\section{Conclusion}
Although the notion of puncturable signatures has been proposed before, this is the first work that makes it efficient enough to be deployed
in practice. We proposed a construction approach based on Bloom filter, whose puncturing operation only involves a small number of efficient computations (e.g. hashing), which outperforms previous schemes by orders of magnitude. Next, we used puncturable signature to construct practical proof-of-stake blockchain protocol resilient to LRSL attack. Our motivation stems from the observation that LRSL attack can alter transactions history and furthermore hamper the development of proof-of-stake blockchain. Our construction allows to realize practical blockchain protocol, and experiment results show that our scheme performs significantly on communication and computation efficiency.

How to design efficient puncturable signature without Bloom filter is a
worthwhile direction. We believe that puncturable signature will find applications beyond proof-of-stake blockchain protocols.

\bibliographystyle{IEEEtran}
\bibliography{abbrev3,extra}

\end{document}